\documentclass[journal,twocolumn,10pt]{IEEEtran}

\usepackage{amsmath,epsfig}
\usepackage{amssymb}
\usepackage{booktabs}
\usepackage{amsthm}
\usepackage{graphicx}
\usepackage{caption}
\usepackage{subcaption}
\usepackage{rotating}
\usepackage{floatflt} 
\usepackage{paralist} 
\usepackage{algorithm} 
\usepackage{xcolor}
\usepackage{color}
\usepackage{epstopdf}
\usepackage[normalem]{ulem}
\usepackage{float}
\usepackage{flexisym}
\usepackage{cite}
\usepackage{colortbl} 

\newcommand{\mypar}[1]{{\bf #1.}}
\newtheorem{myLem}{Lemma}

\newtheorem{myThm}{Theorem}
\newcommand{\R}{\ensuremath{\mathbb{R}}}
\newcommand{\C}{\ensuremath{\mathbb{C}}}
\DeclareMathOperator{\Id}{I}

\def\a{\mathbf{a}}
\def\aw{\widehat{\a}}

\def\c{\mathbf{c}}

\def\x{\mathbf{x}}
\def\xw{\widehat{\x}}

\def\u{\mathbf{u}}
\def\vw{\mathbf{v}}
\def\e{\mathbf{e}}
\def\ew{\widehat{\e}}
\def\t{\mathbf{t}}
\def\w{\mathbf{w}}
\def\ww{\widehat{\w}}
\def\f{\mathbf{f}}

\def\q{\mathbf{q}}

\def\Kset{\mathcal{K}}

\def\V{\mathcal{V}}
\def\M{\mathcal{M}}
\def\U{\mathcal{U}}
\def\I{\mathcal{I}}
\def\D{\mathcal{D}}
\def\Xw{\widehat{\X}}
\def\Ww{\widehat{\W}}
\def\Ew{\widehat{\E}}
\def\Cw{\widehat{\Co}}
\def\Zw{\widehat{\Z}}

\DeclareMathOperator{\Eig}{\Lambda}
\DeclareMathOperator{\Shrink}{\Theta}

\def\L{\mathcal{L}}
\def\ImP{\boldsymbol{\delta}}
\def\Eset{\mathcal{E}}

\DeclareMathOperator{\ZeroMatrix}{\mathbf{0}}
\DeclareMathOperator{\TV}{TV}
\DeclareMathOperator{\STV}{S}
\DeclareMathOperator{\Adj}{A}
\DeclareMathOperator{\W}{W}
\DeclareMathOperator{\Co}{C}
\DeclareMathOperator{\E}{E}

\DeclareMathOperator{\K}{K}
\DeclareMathOperator{\Pj}{P}
\DeclareMathOperator{\Q}{Q}

\DeclareMathOperator{\Um}{U}
\DeclareMathOperator{\Vm}{V}
\DeclareMathOperator{\X}{X}
\DeclareMathOperator{\Y}{Y}
\DeclareMathOperator{\Z}{Z}

\DeclareMathOperator{\T}{T}
\DeclareMathOperator{\prox}{prox}
\DeclareMathOperator{\proj}{proj}

\newcommand{\DSPG}{$\mbox{DSP}_{\mbox{\scriptsize G}}$}

\begin{document}

\title{ Signal Recovery on Graphs: Variation Minimization}
\author{Siheng~Chen,
  Aliaksei~Sandryhaila,
  Jos\'e~M.~F.~Moura,
  Jelena~Kova\v{c}evi\'c
  \thanks{     
  S. Chen, J. M. F. Moura and J. Kova\v{c}evi\'c are with the Department of Electrical and Computer
    Engineering. J. M. F. Moura and J. Kova\v{c}evi\'c, are by courtesy with the Department of Biomedical Engineering, Carnegie
    Mellon University, Pittsburgh, PA, 15213 USA. Emails:
    \{sihengc,moura,jelenak\}@andrew.cmu.edu. A. Sandryhaila was at Carnegie
    Mellon University and now is with
    HP Vertica, Pittsburgh, PA, 15203 USA. Email: aliaksei.sandryhaila@hp.com.}
}
 \maketitle

\begin{abstract}
  We consider the problem of signal recovery on graphs.  Graphs model
  data with complex structure as signals on a graph. Graph signal
  recovery  recovers one or multiple smooth graph signals
  from noisy, corrupted, or incomplete measurements. We formulate graph signal recovery as an optimization problem, for which we provide a general solution through the alternating direction methods of multipliers. We show how
  signal inpainting, matrix completion, robust principal component
  analysis, and anomaly detection all relate to graph signal recovery and provide corresponding specific solutions and theoretical analysis. We validate the proposed methods on real-world
  recovery problems, including online blog classification, bridge
  condition identification, temperature estimation, recommender
  system for jokes, and expert opinion combination of online blog
  classification.
\end{abstract}


\section{Introduction}
\label{sec:intro}
With the explosive growth of information and communication, signals
are being generated at an unprecedented rate from various sources,
including social networks, citation, biological, and physical
infrastructures~\cite{Jackson:08,Newman:10}. Unlike time-series or
images, these signals have complex, irregular structure, which
requires novel processing techniques leading to the emerging field
of~\emph{signal processing on graphs}.

Signal processing on graphs extends classical discrete signal
processing for time-series to signals with an underlying complex,
irregular
structure~\cite{SandryhailaM:13,SandryhailaM:131,ShumanNFOV:13,HammondVG:11}. The
framework models that structure by a graph and
signals by graph signals, generalizing concepts and tools from
classical discrete signal processing to graph signal processing.
Recent work involves graph-based
filtering~\cite{SandryhailaM:13, SandryhailaM:131, NarangO:2012,NarangO:2013},
graph-based
transforms~\cite{SandryhailaM:13,HammondVG:11,NarangSO:10}, sampling
and interpolation on graphs~\cite{Pesenson:08,NarangGO:13, WangLG:14}, uncertainty principle on graphs~\cite{AgaskarL2013},
semi-supervised classification on
graphs~\cite{ChenCRBGK:13,SandryhailaM:13g,EkambaramFAB:13}, graph
dictionary learning~\cite{ZhangDF:12,ThanouSF:13}, and community
detection on graphs~\cite{ChenO:14}; for a recent review,
see~\cite{SandryhailaM:14}.

Two basic approaches to signal processing on graphs have been
considered, both of which analyze signals with complex, irregular
structure, generalizing a series of concepts and tools from classical
signal processing, such as graph filters, or graph Fourier transform, to
diverse graph-based applications, such as graph signal denoising,
compression, classification, and
clustering~\cite{ShumanNFOV:13,ElmoatazLB:08,ChenSMK:13,DongFVN:14}. The
first is rooted in~\emph{spectral graph theory}~\cite{Chung:96,BelkinN:03} and builds on the graph
Laplacian matrix. Since the graph Laplacian matrix is restricted to be
symmetric and positive semi-definite, this approach is applicable only
to undirected graphs with real and nonnegative edge weights. The
second approach,~\emph{discrete signal processing on graphs}
(\DSPG)~\cite{SandryhailaM:13,SandryhailaM:131}, is rooted in the
\emph{algebraic signal processing
  theory}~\cite{PueschelM:08,Pueschelm:08b} and builds on the graph
shift operator, which works as the elementary filter that generates
all linear shift-invariant filters for signals with a given
structure. The graph shift operator is the adjacency matrix and
represents the relational dependencies between each pair of
nodes. Since the graph shift operator is not restricted to be symmetric, this approach is
applicable to arbitrary graphs, those with undirected or directed
edges, with real or complex, nonnegative or negative weights.

In this paper, we consider the classical signal processing task of
signal recovery within the framework of \DSPG. Signal recovery
problems in the current literature include image
denoising~\cite{Mallat:09,BuadesCM:05}, signal
inpainting~\cite{Rudin:92,Chan:01,ChambolleA:04}, and
sensing~\cite{Donoho:06,CandesRT:06a}, but are limited to signals with
regular structure, such as time series. We use \DSPG~to deal with
signals with arbitrary structure, including both undirected and
directed graphs. Graph signal recovery attempts to recover one or
multiple graph signals that are assumed to be smooth with respect to
underlying graphs, from noisy, missing, or corrupted measurements. The
smoothness constraint assumes that the signal samples at neighboring
vertices are similar~\cite{SandryhailaM:131}.

We propose a graph signal model, cast graph signal recovery as an
optimization problem, and provide a general solution by using the
alternating direction method of multipliers. We show that many
classical recovery problems, such as signal
inpainting~\cite{Rudin:92,Chan:01,ChambolleA:04}, matrix
completion~\cite{CandesR:09,CandesP:10}, and robust principal
component analysis~\cite{CandesLMW:11,WrightGRM:09}, are related to
the graph signal recovery problem.  We propose theoretical solutions
and new algorithms for graph signal inpainting, graph signal matrix
completion, and anomaly detection of graph signals, all applicable to
semi-supervised classification, regression, and matrix
completion. Finally, we validate the proposed methods on real-world
recovery problems, including online blog classification, bridge
condition identification, temperature estimation, recommender system,
and expert opinion combination.

\mypar{Previous work} We now briefly review existing work related to recovery problems.~\emph{Image denoising} recovers an image from noisy
observations. Standard techniques include Gaussian smoothing, Wiener
local empirical filtering, and wavelet thresholding methods
(see~\cite{BuadesCM:05} and references therein).  \emph{Signal
  inpainting} reconstructs lost or deteriorated parts of signals,
including images and videos. Standard techniques include total
variation-based
methods~\cite{Rudin:92,Chan:01,ChambolleA:04,ElmoatazLB:08}, image
model-based methods~\cite{ChanS:05}, and sparse
representations~\cite{MairalES:08}.  \emph{Compressed sensing}
acquires and reconstructs signals by taking only a limited number of
measurements~\cite{CandesLMW:11,WrightGRM:09}. It assumes that signals
are sparse and finds solutions to underdetermined linear systems by
$\ell_1$ techniques.  \emph{Matrix completion} recovers the entire
matrix from a subset of its entries by assuming that the matrix is of
low rank. It was originally proposed in~\cite{CandesR:09} and
extensions include a noisy version
in~\cite{CandesP:10,KeshavanMO:10} and decentralized algorithms via graphs~\cite{MardaniMG:13}.~ \emph{Robust principal component
  analysis} recovers a low-rank matrix from corrupted
measurements~\cite{CandesLMW:11,WrightGRM:09}; it separates an image
into two parts: a smooth background and a sparse foreground. In
contrast to principal component analysis, it is robust to grossly
corrupted entries.

Existing work related to signal recovery based on spectral graph theory includes: 1) interpolation of bandlimited graph signals to recover bandlimited graph signals from a set with specific properties, called the uniqueness set~\cite{Pesenson:08, NarangGO:13}. Extensions include the sampling theorem on graphs~\cite{AnisGO:14} and fast distributed algorithms~\cite{WangLG:14}; and 2) smooth regularization on graphs to recover smooth graph signals from random samples~\cite{ZhuGL:03,ZhouS:04,BelkinNS:06}.

\mypar{Contributions}  The contributions of the
  paper are as follows: 
\begin{itemize}
\item a novel algorithm for the general recovery problem that unifies
  existing algorithms, such as signal inpainting, matrix
  completion, and robust principal component analysis;
\item a novel graph signal inpainting algorithm with analysis of the
  associated estimation error;
\item a novel graph signal matrix completion algorithm with analysis
  of a theoretical connection between graph total variation and
  nuclear norm; and
\item novel algorithms for anomaly detection of graph signals with
  analysis of the associated detection accuracy, and 
  for robust graph signal inpainting.
\end{itemize}

\mypar{Outline of the paper} Section~\ref{sec:background} formulates
the problem and briefly reviews \DSPG, which lays the foundation for
this paper; Section~\ref{sec:gsr} describes the proposed solution for
a graph signal recovery problem. Sections that follow study
three subproblems; graph signal inpainting in Section~\ref{sec:gsi},
graph signal matrix completion in Section~\ref{sec:gsmc}, and anomaly
detection of graph signals in Section~\ref{sec:ad}. Algorithms are
evaluated in Section~\ref{sec:results} on real-world recovery
problems, including online blog classification, bridge condition
identification, temperature estimation, recommender system for jokes and expert opinion
combination. Section~\ref{sec:conclusions} concludes the paper and
provides pointers to future directions.

\section{Discrete Signal Processing on Graphs}
\label{sec:background}

We briefly review relevant concepts from \DSPG; for more details,
see~\cite{SandryhailaM:13,SandryhailaM:131}.  \DSPG\ is a theoretical
framework that generalizes classical discrete signal processing from
regular domains, such as lines and rectangular lattices, to arbitrary,
irregular domains commonly represented by graphs, with applications in
signal compression, denoising and classification, semi-supervised
learning and data recovery~\cite{ChenCRBGK:13,ChenSLWMRBGK:14,ChenSMK:13, ChenSMK:14a}.

\mypar{Graph shift} In \DSPG, signals are represented by a graph $G =
(\V,\Adj)$, where $\V = \{v_0,\ldots, v_{N-1}\}$ is the set of nodes,
and $\Adj \in \C^{N \times N}$ is the~\emph{graph shift}, or a
weighted adjacency matrix that represents the connections of the graph
$G$, either directed or undirected. The $n$th signal element
corresponds to the node $v_n$, and the edge weight $\Adj_{n,m}$ between
nodes $v_n$ and $v_m$ is a quantitative expression of the underlying
relation between the $n$th and the $m$th signal samples, such as a
similarity, a dependency, or a communication pattern.

\mypar{Graph signal} \
Given the graph representation $G = (\V,\Adj)$, a \emph{graph signal}
is defined as a map on the graph that assigns the signal samples $x_n\in\C$ to the
node $v_n$.  Once the node order is fixed, the graph
signal can also be written as a vector
\begin{equation}
\label{eq:graph_signal}
  \x \ = \ \begin{bmatrix}
 x_0, x_1, \ldots, x_{N-1}
\end{bmatrix}^T \in \C^N.
\end{equation}

\mypar{Graph Fourier transform} In general, a Fourier transform
corresponds to the expansion of a signal into basis functions that are
invariant to filtering. This invariant basis is the eigenbasis of the graph shift
$\Adj$ (or, if the complete eigenbasis does not exist, the Jordan
eigenbasis of $\Adj$~\cite{SandryhailaM:13}).  For simplicity, assume that $\Adj$ has a complete eigenbasis, and the
spectral decomposition of $\Adj$ is~\cite{VetterliKG:12},
\begin{equation}
  \label{eq:eigendecomposition}
  \Adj=\Vm\Eig\Vm^{-1},
\end{equation}
where the eigenvectors of $\Adj$ form the columns of matrix $\Vm$, and
$\Eig\in\C^{N\times N}$ is the diagonal matrix of the corresponding
eigenvalues $\lambda_0, \, \lambda_1, \, \ldots, \, \lambda_{N-1}$ of $\Adj$.  The
\emph{graph Fourier transform} of a graph
signal~\eqref{eq:graph_signal} is then 
\begin{equation}
  \label{eq:graph_FT}
  \widetilde{\x} = \begin{bmatrix} \widetilde{x}_0 & \widetilde{x}_1 & \ldots & \widetilde{x}_{N-1} \end{bmatrix}^T = \Vm^{-1} \x,
\end{equation}
where $\widetilde{x}_n$ in~\eqref{eq:graph_FT} represents the signal's
expansion in the eigenvector basis and forms the frequency content of
the graph signal $\x$.  The \emph{inverse graph Fourier transform}
reconstructs the graph signal from its frequency content by combining the
graph frequency components weighted by the coefficients of the
signal's graph Fourier transform,
\begin{equation*}
  \label{eq:graph_FT_inverse}
\x = \Vm \widetilde{\x}.
\end{equation*}

\mypar{Variation on graphs} Signal smoothness is a qualitative
characteristic that expresses how much signal samples vary with
respect to the underlying signal representation domain.  To quantify it,  \DSPG\ uses the graph total variation based on the~\emph{$\ell_p$-norm}~\cite{SandryhailaM:131},
\begin{equation}
  \label{eq:GTVQ}
  \STV_{p} (\x) = \left|\left|\x - \frac{1}{|\lambda_{\max}{(\Adj)}|} \Adj \x \right|\right|_p^p,
\end{equation}
where $\lambda_{\max}{(\Adj)}$ denotes the eigenvalue of $\Adj$ with
the largest magnitude.\footnote{ For simplicity, throughout this paper
  we assume that the graph shift $\Adj$ has been normalized to satisfy
  $|\lambda_{\rm max}(\Adj)| = 1$.} We normalize the graph shift to
guarantee that the shifted signal is properly scaled with respect to the original one. When $p = 2$, we call~\eqref{eq:GTVQ} the~\emph{quadratic form of graph total variation}.

%
%
%

We  summarize the notation in Table~\ref{table:parameters}.

\begin{table}[h]
  \footnotesize
  \begin{center}
    \begin{tabular}{@{}lll@{}}
      \toprule
      {\bf Symbol}  & {\bf Description} & {\bf Dimension}\\
      \midrule \addlinespace[1mm]
      $ \x $ &  graph signal &  $N$\\
      $ \X $ &  matrix of graph signal  &  $N \times L$\\
      $ \Adj $ &  graph shift &  $N \times N$\\
	  $ \widetilde{\Adj } $ & $ (\Id-\Adj)^*(\Id-\Adj)$ &  $N \times N$ \\	  
	  $ \M $ &  accessible indices &  \\
	  $ \U $ &  inaccessible indices &  \\
	  $\x_\M $ & accessible elements in $\x$& \\	  
	  $\X_\M$  & accessible elements in $\X$ & \\	
      $\Shrink(\cdot)$	& element-wise shrinkage function defined in~\eqref{eq:shrinkage} & \\
      $\D(\cdot)$	& singular-value shrinkage function defined in~\eqref{eq:shrinkage_singular} & \\
      \bottomrule
    \end{tabular}
  \end{center}
   \caption{\label{table:parameters}   Key notation used in the
     paper.} 
\end{table}

\section{Graph Signal Recovery}
\label{sec:gsr}
We now formulate the general recovery problem for graph signals to
unify multiple signal completion and denoising problems and generalize
them to arbitrary graphs.  In the sections that follow, we consider
specific cases of the graph signal recovery problem, propose
appropriate solutions, and discuss their implementations and
properties.


Let $\x^{(\ell)} \in \C^N$, $1\leq \ell \leq L$, be graph signals residing on the graph $G = (\V,\Adj)$, and let $\X$ be the $N\times L$ matrix of graph signals,
\begin{equation}
\label{eq:Matrix}
\X = \begin{bmatrix} \x^{(1)} & \x^{(2)} & \ldots & \x^{(L)} \end{bmatrix}.
\end{equation}

Assume that we do not know these signals exactly, but for each signal we have
a corresponding measurement $\t^{(\ell)}$.
Since each $\t^{(\ell)}$ can be corrupted by noise and outliers,
we consider the $N\times L$ matrix of measurements to be
\begin{equation}
  \label{eq:Measurement}
  \T = \begin{bmatrix} \t^{(1)} & \t^{(2)} &\ldots & \t^{(L)} \end{bmatrix} = \X + \W + \E,
\end{equation}
where matrices $\W$ and $\E$ contain noise and outliers,
respectively. Note that an outlier is an observation point that is
distant from other observations, which may be due to variability in
the measurement. We assume that the noise coefficients in $\W$ have
small magnitudes, i.e., they can be upper-bounded by a small value,
and that the matrix $\E$ is sparse, containing few non-zero
coefficients of large magnitude. Furthermore, when
certain nodes on a large graph are not accessible, the measurement
$\t^{(\ell)}$ may be incomplete.  To reflect this, we denote the sets
of indices of accessible and inaccessible nodes as $\M$ and $\U$,
respectively. Note that inaccessible nodes denote that values on those
nodes are far from the ground-truth because of corruption, or because we
do not have access to them.

Signal recovery from inaccessible measurements requires additional
knowledge of signal properties. In this work, we make the following
assumptions: (a) the signals of interest $\x^{(\ell)}$, are smooth
with respect to the representation graph $G = (\V,\Adj)$; we express
this by requiring that the variation of recovered signals be small;
(b) since the signals of interest $\x^{(\ell)}$ are all supported on
the same graph structure, we assume that these graph signals are
similar and provide redundant information; we express this by
requiring that the matrix of graph signals $\X$ has low rank; (c) the
outliers happen with a small probability; we express this by requiring that
the matrix $\E$ be sparse; and (d) the noise has small magnitude;
we express this by requiring that the matrix $\W$ be upper-bounded. We
thus formulate the problem as follows:
\begin{eqnarray}
  \label{eq:Opt}
  \Xw, \Ww, \Ew &=& \arg \min_{\X, \W, \E}\  \alpha \STV_2(\X) + \beta~{\rm rank} (\X) + \gamma \left\|\E\right\|_0,
\nonumber  \\
  \\
  \label{eq:Opt_cond1}
  \text{subject to}&& \left\|\W \right\|_F^2 ~\leq~ \epsilon^2,\\
  \label{eq:Opt_cond2}
  && \T_\M = (\X + \W + \E)_\M,
\end{eqnarray}
where $\Xw, \Ww, \Ew$ denote the optimal solutions of the graph signal
matrix, the noise matrix, and the outlier matrix, respectively,
$\epsilon$ controls the noise level, $\alpha, \beta, \gamma$ are
tuning parameters, and
\begin{eqnarray}
\label{eq:quad}
  \STV_2 (\X) &=& \sum_{\ell=1}^L  {\STV_2 (\x^{(\ell)})} = \left\| \X - \Adj \X \right\|_F^2,
\end{eqnarray}
where $\left\|\cdot \right\|_F$ denotes the Frobenius norm and represents
the cumulative quadratic form of the graph total
variation~\eqref{eq:GTVQ} for all graph signals, and
$\left\|\E\right\|_0$ is the $\ell_0$-norm that is defined as the
number of nonzero entries in $\E$.  The general problem~\eqref{eq:Opt}
recovers the graph signal matrix~\eqref{eq:Matrix} from the noisy
measurements~\eqref{eq:Measurement}, possibly when only a subset of
nodes is accessible.

Instead of using the graph total variation based on $\ell_1$ norm~\cite{SandryhailaM:131}, we use the quadratic form of the graph total variation~\eqref{eq:GTVQ} for two reasons. First, it is computationally
easier to optimize than the $\ell_1$-norm based graph total variation. Second, the $\ell_1$-norm based graph total variation, which penalizes less transient changes than the quadratic
form, is good at separating smooth from non-smooth parts of graph
signals; the goal here, however, is to force graph signals at each
node to be smooth. We thus use the quadratic form of the graph total
variation in this paper and, by a slight abuse of notation, call it
graph total variation for simplicity.

The minimization problem~\eqref{eq:Opt} with
conditions~\eqref{eq:Opt_cond1} and~\eqref{eq:Opt_cond2} reflects all
of the above assumptions: (a) minimizing the graph total variation
$\STV_2(\X)$ forces the recovered signals to be smooth and to lie in
the subspace of ``low'' graph frequencies~\cite{SandryhailaM:131}; (b)
minimizing the rank of $\X$ forces the graph signals to be similar and
provides redundant information; (c) minimizing the $\ell_0$-norm
$\left\|\E \right\|_0$ forces the outlier matrix to have few non-zero
coefficients; (d) condition~\eqref{eq:Opt_cond1} captures the
assumption that the coefficients of $\W$ have small magnitudes; and
(e) condition~\eqref{eq:Opt_cond2} ensures that the solution coincides
with the measurements on the accessible nodes.

Unfortunately, solving~\eqref{eq:Opt} is hard because of the rank and
the $\ell_0$-norm~\cite{DonohoE:03, CandesP:09}. To solve it
efficiently, we relax and reformulate~\eqref{eq:Opt} as follows:
\begin{eqnarray}
  \nonumber
  \Xw, \Ww, \Ew &=& \arg \min_{\X, \W, \E}\ \alpha \STV_2(\X) + \beta \left\|\X \right\|_* + \gamma \left\|\E \right\|_1,\\
  \label{eq:C_Opt} \\
   \label{eq:C_Opt_cond1}
  \text{subject to}
  && \left\|\W \right\|_F^2 ~\leq~ \epsilon^2,\\
     \label{eq:C_Opt_cond2}
  && \T_\M = (\X + \W + \E)_\M.
\end{eqnarray}
In~\eqref{eq:C_Opt}, we replace the rank of $\X$ with the nuclear
norm, $\left\|\X\right\|_*$, defined as the sum of all the singular
values of $\X$, which still promotes low
rank~\cite{CandesR:09,CandesP:10}. We further replace the
$\ell_0$-norm of $\E$ with the $\ell_1$-norm, which still promotes
sparsity of $\E$~\cite{DonohoE:03,CandesP:09}. The minimization
problems~\eqref{eq:Opt} and~\eqref{eq:C_Opt} follow the same
assumptions and promote the same properties, but~\eqref{eq:Opt} is an
ideal version, while~\eqref{eq:C_Opt} is a practically feasible
version, because it is a convex problem and thus easier to solve. We
call~\eqref{eq:C_Opt} the~\emph{graph signal recovery} (GSR) problem;
see Table~\ref{table:algorithms}.

To solve~\eqref{eq:C_Opt} efficiently, we use the alternating direction method of
multipliers (ADMM)~\cite{BoydPCPE:11}. ADMM is an algorithm that is
intended to take advantage of both the decomposability and the superior
convergence properties of the method of multipliers. In ADMM, we first
formulate an augmented function and then iteratively update each
variable in an alternating or sequential fashion, ensuring the convergence of the method~\cite{BoydPCPE:11}. We leave the derivation to the Appendix, and summarize the implementation in Algorithm~\ref{alg:GSR}. 
\begin{algorithm}[h]
  \footnotesize
  \caption{\label{alg:GSR} Graph Signal Recovery}
  \begin{tabular}{@{}ll@{}}
    \addlinespace[1mm]
   {\bf Input} 
      & $\T$  ~~~ matrix of measurements \\
     {\bf Output}  
      & $\Xw$  ~~~ matrix of graph signals\\
      & $\Ww$  ~~~matrix of outliers \\
      & $\Ew$  ~~~~matrix of noise \\
    \addlinespace[2mm]
    {\bf Function} 
    & {\bf GSR($\T$) } \\
    & while the stopping criterion is not satisfied  \\ 
    & ~~while the stopping criterion is not satisfied  \\ 
    &~~~~$\X \leftarrow \D_{\beta\eta^{-1}}{ \left( \X + t (\T-\W-\E-\Co- \eta^{-1} (\Y_1+\Y_2) -\Z) \right)}$  \\
    & ~~end \\
    & ~~$\W \leftarrow \eta (\T-\X-\E-\Co- \eta^{-1} \Y_1) /(\eta+2)  $  \\
    & ~~while the stopping criterion is not satisfied  \\     
    &~~~~$\E \leftarrow \Shrink_{\gamma \eta^{-1}}{ \left( \X + t (\T-\X-\W-\Co- \eta^{-1} \Y_1  ) \right)}$  \\    
    & ~~end \\  
    & ~~$\Z \leftarrow (\Id+2\alpha\eta^{-1} (\Id - \Adj)^* (\Id - \Adj)    )^{-1}(\X- \eta^{-1} \Y_2)$  \\
    & ~~$\Co_\M \leftarrow 0, \Co_\U \leftarrow (\T-\X-\W-\E- \eta^{-1} \Y_1 )_\U$, \\    
    & ~~$\Y_1 \leftarrow \Y_1 - \eta (\T-\X-\W-\E-\Co)$  \\
    & ~~$\Y_2 \leftarrow \Y_2 - \eta (\X-\Z)$  \\
    & end \\
    & {\bf return} $\Xw \leftarrow \X,  ~\Ww \leftarrow \W,  ~\Ew \leftarrow \E$ \\  
     \addlinespace[1mm]
  \end{tabular}
\end{algorithm}
Note that in Algorithm~\ref{alg:GSR}, $\Y_1, \Y_2$ are Lagrangian multipliers, $\eta$ is pre-defined, the step
size $t$ is chosen from backtracking line search~\cite{BoydV:04}, and
operators $\Shrink_\tau$ and $\D_\tau$ are defined for $\tau\geq 0$ as
follows:  $\Shrink_\tau(\X)$ ``shrinks'' every element of $\X$ by
$\tau$ so that the $(n,m)$th element of $\Shrink_\tau(\X)$ is
\begin{equation}
\label{eq:shrinkage}
\Shrink_\tau(\X)_{n,m} = \begin{cases}
\X_{n,m}-\tau,& \text{when \,\,\,} \X_{n,m}\geq\tau, \\
\X_{n,m}+\tau,& \text{when \,\,\,} \X_{n,m}\leq-\tau, \\
0,& \text{otherwise}.
\end{cases}
\end{equation}
Similarly, $\D_\tau(\X)$ ``shrinks'' the singular values of $\X$,
\begin{equation}
\label{eq:shrinkage_singular}
\D_\tau(\X) = \Um \Shrink_{\tau}{(\Sigma)} \Q^*,
\end{equation}
where $\X = \Um \Sigma \Q^*$ denotes the singular value decomposition
of $\X$~\cite{VetterliKG:12} and $*$ denotes the Hermitian transpose. The following stopping criterion is used in the paper: the difference of the objective function between two consecutive iterations is smaller than $10^{-8}$. The bulk of the computational cost is in the singular value decomposition~\eqref{eq:shrinkage_singular} when updating $\X$, which is also involved in the standard implementation of matrix completion.

We now review several well-known algorithms
for signal recovery, including signal inpainting, matrix completion, and robust principal component analysis, and show how
they can be formulated as special cases of the graph signal recovery
problem~\eqref{eq:Opt}. In Sections~\ref{sec:gsi} and~\ref{sec:gsmc}, we show graph counterparts of the signal inpainting and matrix completion problems by minimizing the graph total variation. In Section~\ref{sec:ad}, we show anomaly detection on graphs, which is inspired by robust principal component analysis.

\mypar{Signal inpainting} Signal inpainting is a process of recovering
inaccessible or corrupted signal samples from accessible samples using
regularization~\cite{Rudin:92,Chan:01,ChambolleA:04,ElmoatazLB:08},
that is, minimization of the signal's total variation. The measurement
is typically modeled as
\begin{equation}
  \label{eq:Signal_inpainting}
  \t = \x + \w \in \R^{N},
\end{equation}
where $\x$ is the true signal, and $\w$ is the noise. Assuming we can
access a subset of indices, denoted as $\M$, the task is then to
recover the entire true signal $\x$, based on the accessible
measurement $\t_\M$. We assume that the true signal $\x$ is smooth,
that is, its variation is small. The variation is expressed by a
\emph{total variation} function
\begin{equation}
\label{eq:TV}
\TV{(\x)} \ = \ \sum_{i=1}^{N} | x_i - x_{i-1~{\rm mod}~N}|.
\end{equation}
We then recover the signal $\x$ by solving the following optimization
problem:
\begin{eqnarray}
  \label{eq:TV_Opt}
  \xw &=& \arg \min_\x \, \TV (\x), \\
  \label{eq:TV_Opt_cond}
  \text{subject to}&& \left\|(\x - \t)_\M\right\|_2^2 ~\leq~ \epsilon^2.
\end{eqnarray}
The condition~\eqref{eq:TV_Opt_cond} controls how well the accessible
measurements are preserved. As discussed in
Section~\ref{sec:background}, both the $\ell_1$ norm based graph total
variation and the quadratic form of the graph total
variation~\eqref{eq:GTVQ} are used. Thus, \eqref{eq:TV_Opt} is a
special case of~\eqref{eq:C_Opt} when the graph shift is the cyclic
permutation matrix, $\alpha = 1$, $L=1$, $\beta=\gamma=0$, $\E=0$, and
 conditions~\eqref{eq:C_Opt_cond1} and~\eqref{eq:C_Opt_cond2} are combined
into the single condition~\eqref{eq:TV_Opt_cond}; see
Table~\ref{table:algorithms}.

\mypar{Matrix completion}
\label{sec:matrixcompletion}
Matrix completion recovers a matrix given a subset of
its elements, usually, a subset of rows or columns.  Typically, the
matrix has a low rank, and the missing part is recovered
through rank minimization~\cite{CandesR:09,CandesP:10,KeshavanMO:10}. The matrix is modeled as
\begin{equation}
\label{eq:Matrix_completion}
\T = \X + \W \in \R^{N \times L},
\end{equation}
where $\X$ is the true matrix and $\W$ is the noise. Assuming we can
access a subset of indices, denoted as $\M$, the matrix $\X$ is
recovered from~\eqref{eq:Matrix_completion} as the solution with the
lowest rank:
\begin{eqnarray}
\label{eq:MC_Opt}
	\Xw &=& \arg \min_\X\, \left\| \X  \right\|_* , \\
\label{eq:MC_Opt_cond}
\text{subject to}&& \left\|(\X - \T)_\M \right\|_2^2 ~\leq~ \epsilon^2;
\end{eqnarray}
this is a special case of~\eqref{eq:C_Opt} with $\alpha=\gamma=0$,
$\beta = 1$, $\E=0$, and conditions~\eqref{eq:C_Opt_cond1} and
\eqref{eq:C_Opt_cond2} are combined into the single condition \eqref{eq:MC_Opt_cond}; see
Table~\ref{table:algorithms}. This also means that the values in the matrix are associated with a graph that is represented by the identity matrix, that is, we do not have any prior information about the graph structure.

\mypar{Robust principal component analysis} Similarly to matrix
completion, robust principal component analysis is used for recovering
low-rank matrices.  The main difference is the assumption that all
matrix elements are measurable but corrupted by
outliers~\cite{CandesLMW:11,WrightGRM:09}.  In this setting, the
matrix is modeled as
\begin{equation}
\label{eq:Matrix_PCA}
\T = \X + \E \in \R^{N \times L},
\end{equation}
where $\X$ is the true matrix, and $\E$ is a sparse matrix of outliers.

The matrix $\X$ is recovered from~\eqref{eq:Matrix_PCA} as the
solution with the lowest rank and fewest outliers:
\begin{eqnarray*}
\label{eq:RPCA_Opt}
   	\Xw, \Ew&=& \arg \min_{\X,\E}\, \beta \left\|{\X} \right\|_* + \gamma \left\|\E \right\|_1, \\
   \label{eq:RPCA_Opt_cond}
	\text{subject to}&& \T = \X + \E;
\end{eqnarray*}
this is a special case of~\eqref{eq:C_Opt} with $\alpha=\epsilon=0$,
$\W=0$, and $\M$ contains all the indices; see
Table~\ref{table:algorithms}. Like before, this also means that the matrix is associated with a graph that is represented by the identity matrix, that is, we do not have any prior information about the graph structure.

\section{Graph Signal Inpainting}
\label{sec:gsi}
We now discuss in detail the problem of signal inpainting on graphs. Parts of this section have appeared in~\cite{ChenSLWMRBGK:14},
and we include them here for completeness.

As discussed in Section~\ref{sec:gsr}, signal
inpainting~\eqref{eq:TV_Opt} seeks to recover the missing entries of
the signal $\x$ from incomplete and noisy measurements under the
assumption that two consecutive signal samples in $\x$ have similar
values. Here, we treat $\x$ as a graph signal that is smooth with
respect to the corresponding graph. We thus update the signal
inpainting problem~\eqref{eq:TV_Opt}, and formulate the~\emph{graph
  signal inpainting} problem\footnote{If we build a graph to model a
  dataset by representing signals or images in the dataset as nodes
  and the similarities between each pair of nodes as edges, the
  corresponding labels or values associated with nodes thus form a
  graph signal, and the proposed inpainting algorithm actually tackles
  semi-supervised learning with graphs~\cite{Zhu:05}.} as
\begin{eqnarray}
\label{eq:GTV_Opt}
	\xw &=& \arg \min_\x \, \STV_{2}(\x), \\
\label{eq:GTV_Opt_cond}
\text{subject to}&& \left\|(\x - \t)_\M\right\|_2^2 ~\leq~ \epsilon^2;
\end{eqnarray}
this is a special case of~\eqref{eq:C_Opt} with $L = 1, \beta =
\gamma = 0$; see Table~\ref{table:algorithms}.

\mypar{Solutions}
In general, graph signal inpainting~\eqref{eq:GTV_Opt} can be solved
 by using Algorithm~\ref{alg:GSR}.  However, in special cases, there exist closed-form solutions that do not require iterative
algorithms.

\subsubsection{Noiseless inpainting}
Suppose that the measurement $\t$ in~\eqref{eq:Signal_inpainting} does not
contain noise.  In this case, $\w=0$, and we solve~\eqref{eq:GTV_Opt}
for $\epsilon=0$:
\begin{eqnarray}
\label{eq:GTVM}
\xw &=& \arg \min_\x \, \STV_2 (\x), \\ \nonumber
\label{eq:GTVM_cond}
\text{subject to}&& \x_\M = \t_\M.
\end{eqnarray}
We call the problem~\eqref{eq:GTVM}~\emph{graph signal inpainting via
  total variation minimization} (GTVM)~\cite{ChenSLWMRBGK:14}.

Let  $\widetilde{\Adj} = (\Id-\Adj)^*(\Id-\Adj)$.
By reordering nodes, write $\widetilde{\Adj}$ in  block form as
\begin{equation}
\label{eq:wAdj}
\nonumber
\widetilde{\Adj} = \begin{bmatrix}
\widetilde{\Adj}_{\M\M} & \widetilde{\Adj}_{\M\U}  \\
\widetilde{\Adj}_{\U\M} & \widetilde{\Adj}_{\U\U}
\end{bmatrix},
\end{equation}
and set the derivative of~\eqref{eq:GTVM} to
$0$; the closed-form solution is
\begin{equation}
\label{eq:solution}
\nonumber
\widehat{\x} =
\begin{bmatrix}
\t_\M \\
-\widetilde{\Adj}_{\U\U}^{-1}\widetilde{\Adj}_{\U\M} \t_\M
\end{bmatrix}.
\end{equation}
When $\widetilde{\Adj}_{\U\U}$ is not invertible, a pseudoinverse should be used.

\subsubsection{Unconstrained inpainting}
The graph signal inpainting~\eqref{eq:GTV_Opt} can be formulated as an
unconstrained problem by merging condition~\eqref{eq:GTV_Opt_cond}
with the objective function:
\begin{equation}
  \label{eq:GTVR}
	\xw = \arg \min_\x \, \left\|(\x - \t)_\M \right\|_2^2 + \alpha \STV_{2} (\x),
\end{equation}
where the tuning parameter $\alpha$ controls the trade-off between the
two parts of the objective function.  We call~\eqref{eq:GTVR}
the~\emph{graph signal inpainting via total variation regularization}
(GTVR). GTVR is a convex quadratic problem that has a
closed-form solution.  Setting the derivative of~\eqref{eq:GTVR} to
zero, we obtain the closed-form solution
\begin{equation}
\nonumber
 \xw = \left(  \begin{bmatrix}
 \Id_{\M \M}  & 0 \\ 0 & 0 \end{bmatrix}
 + \alpha \widetilde{\Adj} \right)^{-1}
 \begin{bmatrix}
  \t_\M \\
 \ZeroMatrix
 \end{bmatrix},
\end{equation}
where $ \Id_{\M \M}$ is an identity matrix. When the term in parentheses is not invertible, a pseudoinverse should be adopted.

\mypar{Theoretical analysis}
Let $\x^0$ denote the \emph{true} graph signal that we are trying to
recover.  Assume that $\STV_{2}{(\x^0)} = \eta^2$ and $\x^0$
satisfies~\eqref{eq:GTV_Opt_cond}, so that $\left\|\x^0_\M - \t_\M\right\|_2^2
~\leq~ \epsilon^2.$ Similarly to~\eqref{eq:wAdj}, we write $\Adj$ in a
block form as
\begin{equation}
\nonumber
{\Adj} \ = \ \begin{bmatrix}
{\Adj}_{\M\M} & {\Adj}_{\M\U}  \\
{\Adj}_{\U\M} & {\Adj}_{\U\U}
\end{bmatrix}.
\end{equation}
The following results, proven in~\cite{ChenSLWMRBGK:14}, establish an
upper bound on the error of the solution to graph signal
inpainting~\eqref{eq:GTV_Opt}.

\begin{myLem}
\label{lem:bound}
The error $\left\|\x^0 - \xw\right\|_2$ of the solution $\xw$ to the graph signal
inpainting problem~\eqref{eq:GTV_Opt} is bounded as
\begin{eqnarray}
\nonumber
\left\|\x^0 - \xw\right\|_2 &\leq& \frac{q}{2}\left\|(\x^0 - \xw)_\U\right\|_2 + p|\epsilon| +|\eta|,
\end{eqnarray}
where
$$
p = \left|\left| \begin{bmatrix} \Id_{\M\M} + \Adj_{\M\M} \\
      \Adj_{\U\M} \end{bmatrix} \right|\right|_2,\,\,\, q =
\left|\left| \begin{bmatrix} \Adj_{\M\U} \\ \Id_{\U\U} +
      \Adj_{\U\U} \end{bmatrix} \right|\right|_2,
$$
and $\left\|\cdot\right\|_2$ for matrices denotes the spectral norm.
\end{myLem}

\begin{myThm}
  \label{thm:bound}
  If $q < 2$, then the error on the
  inaccessible part of the solution $\xw$ is bounded as
\begin{eqnarray}
  \nonumber
  \left\|(\x^0 - \xw)_\U\right\|_2 &\leq& \frac{ 2p|\epsilon| + 2 |\eta|}{2 - q}.
\end{eqnarray}
\end{myThm}

The condition $q < 2$ may not hold for some
matrices; however, if $\Adj$ is symmetric, we have $q ~\leq~ \left\|\Id +
\Adj \right\|_2$ $~\leq~ \left\|\Id\right\|_2 + \left\|\Adj\right\|_2 = 2$, since $\left\|\Adj\right\|_2 =
1$. Since $q$ is related to the size of the inaccessible part, when we take a larger number of measurements, $q$ becomes smaller, which leads to a tighter upper bound. Also, note that the upper bound is related to the smoothness of
the true graph signal and the noise level of the accessible
part. A central assumption of any inpainting technique is that the
true signal $\x^0$ is smooth. If this assumption does not hold, then
the upper bound is large and useless. When the noise level of the
accessible part is smaller, the measurements from the accessible part
are closer to the true values, which leads to a smaller estimation
error.

%

\section{Graph Signal Matrix Completion}
\label{sec:gsmc}
We now consider graph signal matrix completion---another
important subproblem of the general graph signal recovery problem~\eqref{eq:Opt}.

As discussed in
Section~\ref{sec:gsr}, matrix
completion seeks to recover missing entries of matrix $\X$ from the
incomplete and noisy measurement matrix~\eqref{eq:Matrix_completion}
under the assumption that $\X$ has low rank.  Since we view $\X$ as a
matrix of graph signals (see~\eqref{eq:Matrix}), we also assume that
the columns of $\X$ are smooth graph signals.  In this case, we update
the matrix completion problem~\eqref{eq:MC_Opt} and formulate the
\emph{graph signal matrix completion} problem as
  \begin{eqnarray}
  \label{eq:MCG_Opt}
   	\Xw &=& \arg \min_\X \, \STV_{2}(\X) + \beta \left\|\X\right\|_*, \\ \nonumber
   \label{eq:MCG_Opt_cond} 
	\text{subject to}&& \left\|(\X - \T)_\M \right\|_F^2 ~\leq~ \epsilon^2;
   \end{eqnarray}
   this is a special case of~\eqref{eq:C_Opt} with $ \alpha =
   1, \gamma = 0$; see Table~\ref{table:algorithms}.

\mypar{Solutions}
\label{sec:alternative_sol}
In addition to Algorithm~\ref{alg:GSR} that can be used to solve the
graph signal matrix completion problem~\eqref{eq:MCG_Opt}, there exist
alternative approaches that we discuss next.

\setcounter{subsubsection}{0}
\subsubsection{Minimization}
Here we consider the noise-free case. Suppose the measurement matrix $\T$ in~\eqref{eq:Matrix_completion}
does not contain noise.  We thus solve~\eqref{eq:MCG_Opt} for
$\W=0$ and $\epsilon=0$,
  \begin{eqnarray}
  \label{eq:MCM_Opt}
   	\Xw &=& \arg \min_\X\, \STV_{2}(\X) + \beta \left\|\X\right\|_*, \\
   \nonumber
	\text{subject to}&&  \X_\M = \T_\M.
   \end{eqnarray}
   We call~\eqref{eq:MCM_Opt}~\emph{graph signal matrix
     completion via total variation minimization} (GMCM). This is a
   constrained convex problem that can be solved by projected
   generalized gradient descent~\cite{BoydV:04}. We first split the
   objective function into two components, a convex, differential
   component, and a convex, nondifferential component; based on these two
   components, we formulate a proximity function and then solve it
   iteratively.  In each iteration, we solve the proximity function
   with an updated input and project the result onto the feasible
   set. To be more specific, we split the objective
   function~\eqref{eq:MCM_Opt} into a convex, differentiable component
   $\STV_{2}(\X)$, and a convex, nondifferential component $\beta
   \left\|\X\right\|_*$. The proximity function is then defined as
\begin{eqnarray}
	\prox_{t} (\X) 	& = & \arg \min_{\Z} \frac{1}{2t} \left\|\X-\Z\right\|^2 + \beta \left\|\X\right\|_*
	\nonumber \\  \nonumber
	 & = & \D_{t\beta} (\Z),
\end{eqnarray}
where $\D(\cdot)$ is defined in~\eqref{eq:shrinkage_singular}.  In
each iteration, we first solve for the proximity function and project the
result onto the feasible set as
\begin{equation}
  \nonumber
  \X \leftarrow {\proj} \left( \prox_t  \left( \X - t \nabla \STV_{2}(\X) \right) \right),
\end{equation}
where $t$ is the step size that is chosen from the backtracking line
search~\cite{BoydV:04}, and $\proj (\X)$ projects $\X$ to the feasible
set so that the $(n,m)$th element of $\proj(\X)$ is
\begin{equation}
  \label{eq:proj}
  \nonumber
  \proj(\X)_{n,m} = \begin{cases}
    \T_{n,m},& \text{when \,\,\,} (n,m) \in \M, \\
    \X_{n,m},& \text{when \,\,\,} (n,m) \in \U. \\
\end{cases}
\end{equation}
For implementation details, see Algorithm~\ref{alg:GMCM}. The bulk of the computational cost of Algorithm~\ref{alg:GMCM} is in the singular value decomposition~\eqref{eq:shrinkage_singular} when updating $\X$, which is also involved in the standard implementation of matrix completion.
\begin{algorithm}[t]
  \footnotesize
  \caption{\label{alg:GMCM} Graph Signal Matrix Completion via Total
    Variation Minimization}
  \begin{tabular}{@{}ll@{}}
    \addlinespace[1mm]
    {\bf Input}
    &  $\T$  ~~~ matrix of measurements\\
    &  $\Xw$  ~~~ matrix of graph signals\\
    \addlinespace[2mm]
    {\bf Function} & {\bf GMCM($\T$) } \\
    & initialize  $\X$, such that  $\X_\M = \T_\M$ holds \\
    & while the stopping criterion is not satisfied   \\
    & ~~~Choose step size $t$ from backtracking line search \\
    & ~~~$\X \leftarrow \proj \left( \D_{t\beta} ( \X - 2 t  \widetilde{\Adj} \X )   \right)$ \\
    & end\\
    & {\bf return} $\Xw \leftarrow \X$ \\
    \addlinespace[1mm]
  \end{tabular}
\end{algorithm}

\subsubsection{Regularization}
The graph signal matrix completion~\eqref{eq:MCG_Opt} can be formulated as an unconstrained problem,
\begin{eqnarray}
  \label{eq:GMCR}
\noindent  \Xw &=& \arg \min_\X \, \left\|\X_\M - \T_\M \right\|_F^2 + \alpha \STV_{2}(\X) + \beta \left\|\X\right\|_*.
\end{eqnarray}
We call~\eqref{eq:GMCR}~\emph{graph signal matrix
  completion via total variation regularization} (GMCR). This is an
unconstrained convex problem and can be solved by generalized gradient
descent. Similarly to projected generalized gradient descent, generalized
gradient descent also formulates and solves a proximity function. The
only difference is that generalized gradient descent does not need to
project the result after each iteration to a feasible set. To be more
specific, we split the objective funtion~\eqref{eq:MCM_Opt} into a
convex, differentiable component $\left\|\X_\M - \T_\M \right\|_F^2 + \alpha
\STV_{2}(\X)$, and a convex, non-differential component $\beta
\left\|\X\right\|_*$. The proximity function is then defined as
\begin{eqnarray}
	\nonumber
	\prox_{t} (\X) 	& = & \arg \min_{\Z} \frac{1}{2t} \left\|\X-\Z\right\|^2 + \beta \left\|\X\right\|_*
	\nonumber \\
	 & = & \D_{t\beta} (\Z),
\end{eqnarray}
where $\D(\cdot)$ is defined in~\eqref{eq:shrinkage_singular}.  In
each iteration, we first solve for the proximity function as
\begin{equation}
  \X \leftarrow \prox_t  \left( \X - t \nabla \left(\left\|\X_\M - \T_\M\right\|_F^2 + \alpha \STV_{2}(\X)  \right) \right),
\end{equation}
where $t$ is the step size that is chosen from the backtracking line
search~\cite{BoydV:04}; for implementation details, see
Algorithm~\ref{alg:GMCR}.
\begin{algorithm}[t]
  \footnotesize
  \caption{\label{alg:GMCR} Graph Signal Matrix Completion via Total Variation Regularization}
  \begin{tabular}{@{}ll@{}}
    \addlinespace[1mm]
    {\bf Input}
    &  $\T$  ~~~matrix of measurements\\
    &  $\Xw $  ~~~matrix of graph signals\\
    \addlinespace[2mm]
    {\bf Function} & {\bf GMCR($\T$)} \\
    & initialize  $\X$ \\
    & while the stopping criterion is not satisfied   \\
    & ~~~Choose step size $t$ from backtracking line search \\
    & ~~~$\X \leftarrow \D_{t\beta} \left( \X - 2 t  (\X_\M - \T_\M) - 2 \alpha t  \widetilde{\Adj}  \X  \right)$ \\
    & end\\
    & {\bf return} $\Xw \leftarrow \X$ \\
    \addlinespace[1mm]
  \end{tabular}
\end{algorithm}
The bulk of the computational cost of Algorithm~\ref{alg:GMCR} is in the singular value decomposition~\eqref{eq:shrinkage_singular} when updating $\X$, which is also involved in the standard implementation of matrix completion.

\mypar{Theoretical analysis}
We now discuss properties of the proposed algorithms. The key in
classical matrix completion is to minimize the nuclear norm of a
matrix. Instead of considering general matrices, we only focus on
graph signal matrices, whose corresponding algorithm is to minimize
both the graph total variation and the nuclear norm. We study the
connection between graph total variation and nuclear norm of a matrix
to reveal the underlying mechanism of our algorithm.

 Let $\X$ be a $N \times L$ matrix of rank $r$ with singular
  value decomposition $\X = \Um \Sigma \Q^*$, where $\Um
  = \begin{bmatrix} \u_1 & \u_2 & \ldots & \u_r \end{bmatrix}$, $\Q
  = \begin{bmatrix} \q_1 & \q_2 & \ldots & \q_r \end{bmatrix}$, and
  $\Sigma$ is a diagonal matrix with $ \sigma_i$ along the diagonal, $i = 1, \cdots, r$.
\begin{myLem}
  \label{thm:TVMC}
  \begin{eqnarray}
    \nonumber
    \STV_{2} (\X) = \sum_{i=1}^{r} {\sigma_i^2 \left\| (\Id - \Adj) \u_i \right\|_2^2}.
  \end{eqnarray}
\end{myLem}

\begin{proof}
  \begin{eqnarray*}
    \STV_{2} (\X) & = & \left\| \X - \Adj \X \right\|_F^2
    \ \stackrel{(a)}{=} \ \left\| (\Id - \Adj) \Um \Sigma \Q^*\right\|_F^2,
    \nonumber \\
    & =  & {\rm Tr} \Biggl(\Q \Sigma \Um^*   (\Id - \Adj)^*(\Id - \Adj) \Um \Sigma \Q^* \Biggr),
    \nonumber \\
    & \stackrel{(b)}{=}  & {\rm Tr} \Biggl(   \Sigma \Um^*   (\Id - \Adj)^*(\Id - \Adj) \Um \Sigma \Q^* \Q \Biggr),
    \nonumber \\
& {=}  & \left\| (\Id - \Adj) \Um \Sigma \right\|_F^2
\ \stackrel{(c)}{=}  \  \sum_{i=1}^{r} {\sigma_i^2 \left\| (\Id - \Adj) \u_i \right\|_2^2},
\end{eqnarray*}
where $(a)$ follows from the singular value decomposition; $(b)$
from the cyclic property of the trace operator; and $(c)$ from
$\Sigma$ being a diagonal matrix.
\end{proof}

From Lemma~\ref{thm:TVMC}, we see that graph total variation is
related to the rank of $\X$; in other words, lower rank naturally
leads to smaller graph total variation.

\begin{myThm}
\begin{eqnarray}
  \nonumber
  \STV_{2} (\X)  \leq \STV_2 (\Um) \left\|\X\right\|_*^2.
\end{eqnarray}
\end{myThm}

\begin{proof}
  From~Lemma~\ref{thm:TVMC}, we have
  \begin{eqnarray*}
    \STV_{2} (\X) &  {=} & \left\| (\Id - \Adj) \Um \Sigma \right\|_F^2
    \ \stackrel{(a)}{\leq} \ \left\| (\Id - \Adj) \Um \right\|_F^2 \left\|\Sigma\right\|_F^2,
    \nonumber \\
    & \stackrel{(b)}{\leq} & \left\| (\Id - \Adj) \Um \right\|_F^2 \left\|\Sigma \right\|_*^2
    \ {=} \ \left\| \Um - \Adj \Um \right\|_F^2 \left\|\X \right\|_*^2,
  \end{eqnarray*}
  where $(a)$ follows from the submultiplicativity of the Frobenius norm;
  and $(b)$ from the norm equivalence~\cite{VetterliKG:12}.
\end{proof}
In Theorem~\ref{thm:TVMC}, we see that the graph total variation is
related to two quantities: the nuclear norm of $\X$ and
the graph total variation of the left singular vectors of $\X$. The
first quantity reveals that minimizing the nuclear norm potentially leads to
minimizing the graph total variation.  We can thus rewrite the objective
function~\eqref{eq:MCG_Opt} as
\begin{eqnarray}
  \nonumber
  \STV_{2}(\X) + \beta \left\|\X\right\|_*  & \leq & \STV_{2}(\Um)  \left\|\X\right\|_*^2+ \beta  \left\|\X\right\|_*.
\end{eqnarray}
If the graph shift is built from insufficient information, we just
choose a larger $\beta$ to force the nuclear norm to be small, which
causes a small graph total variation in return. The quantity
$\STV_{2} (\Um)$ measures the smoothness of the left singular vectors of $\X$
on a graph shift $\Adj$; in other words, when the left singular vectors are smooth, the graph signal matrix is also smooth. We can further use this quantity to bound the
graph total variation of all graph signals that belong to a subspace spanned by the left singular vectors.

\begin{myThm}
\label{thm:signal_space_smooth}
  Let a graph signal $\x$ belong to the space spanned by $\Um$,
  that is, $\x = \Um \a$, where $\a$ is the vector of representation
  coefficients. Then,
$$
\STV_2(\x) \leq \STV_2(\Um) \left\|\a\right\|_2^2
$$
\end{myThm}
\begin{proof}
\begin{eqnarray*}
  \STV_{2}(\x) & = & \left\|\x - \Adj \x\right\|_2^2
  \ = \ \left\|(\Id - \Adj) \Um \a\right\|_2^2
  \nonumber \\
  & \stackrel{(a)}{\leq} & \left\|(\Id - \Adj) \Um\right\|_2^2  \left\|\a\right\|_2^2
  \ \stackrel{(b)}{\leq} \ \left\|(\Id - \Adj) \Um\right\|_F^2  \left\|\a\right\|_2^2.
\end{eqnarray*}
where (a) follows from the submultiplicativity of the spectral norm; and
(b) from the norm equivalence~\cite{VetterliKG:12}.
\end{proof}
Theorem~\ref{thm:signal_space_smooth} shows that a graph signal is smooth when it belongs to a subspace spanned by the smooth left singular vectors.

\begin{table*}
  \footnotesize
  \begin{center}
    \begin{tabular}{@{}ll@{}}
      \toprule
      {\bf Graph signal recovery problem} &  {$\Xw, \Ww, \Ew = \arg \min_{\X, \W, \E \in \R^{N \times L}}\  \alpha \STV_2(\X) + \beta \left\|\X\right\|_* + \gamma \left\|\E\right\|_1,$}\\
      &  { subject to $\left\|\W \right\|_F^2 ~\leq~ \epsilon^2, \T_\M = (\X + \W + \E)_\M.$}\\
      \midrule \addlinespace[1mm]
      {\bf Signal inpainting} &  { $L = 1, \beta = 0, \gamma = 0,$ graph shift is the cyclic permutation matrix. }\\
      {\bf Matrix completion} & { $\alpha = 0$, or graph shift is the identity matrix, $\gamma = 0$ }\\
      {\bf Robust principal component analysis} & { $\alpha = 0$, or graph shift is the identity matrix, $\epsilon = 0, \M$ is all the indices in $\T$. } \\
      {\bf Graph signal inpainting} & { $L = 1, \beta = 0, \gamma = 0$ }\\
      {\bf Graph signal matrix completion} & {  $\alpha = 1, \gamma = 0$ }\\
      {\bf Anomaly detection} & { $L = 1, \beta = 0, \M$ is all indices in $\T$ } \\
      {\bf Robust graph signal inpainting} & { $L = 1, \beta = 0$ } \\
      \addlinespace[1mm] \bottomrule
    \end{tabular}
  \end{center}
  \caption{\label{table:algorithms} The table of algorithms.}
\end{table*}

\section{Anomaly Detection}
\label{sec:ad}
We now consider anomaly detection of graph signals, another important
subproblem of the general recovery problem~\eqref{eq:Opt}.

 As discussed in
Section~\ref{sec:gsr}, robust principal component analysis seeks to detect outlier coefficients from a low-rank matrix. Here, anomaly detection of graph signals seeks to  detect outlier coefficients from a smooth graph signal. We assume that the
outlier is sparse and contains few non-zero coefficients of large
magnitude. To be specific, the measurement is modeled as
\begin{equation}
\label{eq:signal_anomaly}
\t = \x + \e \in \R^N,
\end{equation}
where $\x$ is a smooth graph signal that we seek to recover, and the
outlier $\e$ is sparse and has large magnitude on few nonzero
coefficients. The task is to detect the outlier $\e$ from the
measurement $\t$. Assuming that $\x$ is smooth, that is, its variation
is small, and $\e$ is sparse, we propose the ideal optimization problem as follows:
  \begin{eqnarray}
  \label{eq:AD_Opt}
   	\xw, \ew &=& \arg \min_{\x, \e}\,  \left\|\e\right\|_0\\
   \label{eq:AD_Opt_cond}
   \nonumber
	\text{subject to}&&   \STV_{2}(\x) \leq \eta^2,
	\\ 
	\nonumber
	\label{eq:AD_Opt_cond_2}
	&& \t = \x + \e.
   \end{eqnarray}
   To solve it efficiently, instead of dealing with the $\ell_0$ norm, we relax it to the $\ell_1$ norm and reformulate~\eqref{eq:AD_Opt} as follows: \\
\begin{eqnarray}
  \label{eq:rAD_Opt}
   	\xw, \ew &=& \arg \min_{\x, \e}\,  \left\|\e\right\|_1
   	 \\  
   \label{eq:rAD_Opt_cond1}
	\text{subject to}&&   \STV_{2}(\x) \leq \eta^2,
	\\ 
	\label{eq:rAD_Opt_cond2}
	&& \t = \x + \e;
\end{eqnarray}
this is a special case of~\eqref{eq:C_Opt} with $L = 1, \beta = 0, \M$
contains all indices in $\t$, and choosing $\alpha$ properly to ensure
that~\eqref{eq:AD_Opt_cond} holds, see
Table~\ref{table:algorithms}. In Section~\ref{sec:AD_dicussion}, we
show that, under these assumptions, both~\eqref{eq:AD_Opt}
and~\eqref{eq:rAD_Opt} lead to perfect outlier detection.

\mypar{Solutions}
The minimization problem~\eqref{eq:rAD_Opt} is convex, and it is numerically efficient to solve for its optimal solution.  

We further formulate an unconstrained problem as follows:
  \begin{eqnarray}
  \label{eq:AD_Opt_unconstraint}
   	\ew &=& \arg \min_\e \,  \STV_{2}( \t - \e )  + \beta \left\|\e\right\|_1.
   \end{eqnarray}
We call~\eqref{eq:AD_Opt_unconstraint}~\emph{anomaly detection via $\ell_1$ regularization} (AD). In~\eqref{eq:AD_Opt_unconstraint}, we merge conditions~\eqref{eq:rAD_Opt_cond1} and~\eqref{eq:rAD_Opt_cond2} and  move them from the constraint to the objective function. We solve~\eqref{eq:AD_Opt_unconstraint} by using generalized gradient descent, as discussed in Section~\ref{sec:alternative_sol}. For implementation details, see Algorithm~\ref{alg:AD}.

\begin{algorithm}[t]
  \footnotesize
  \caption{\label{alg:AD} Anomaly detection via $\ell_1$ regularization}
  \begin{tabular}{@{}ll@{}}
    \addlinespace[1mm]
    {\bf Input}
    &$\t$~~~input graph signals \\
    {\bf Output}
    &$\ew$~~~outlier signals \\
    \addlinespace[2mm]
    {\bf Function} & {\bf AD$(\x)$ } \\
    & initialize  $\e$ \\
    & while the stopping criterion is not satisfied   \\
    & ~~~Choose step size $t$ from backtracking line search \\
    & ~~~$\e \leftarrow \Shrink_{t\beta} \left( \e - 2 t \widetilde{\Adj}  (\t -\e)  \right)$ \\
    & end\\
    & {\bf return} $\ew  \leftarrow \e$ \\
    \addlinespace[1mm]
  \end{tabular}
\end{algorithm}

\mypar{Theoretical analysis}
\label{sec:AD_dicussion}
Let $\x^0$ be the true graph signal, represented as $\x^0 = \Vm \a^0 = \sum_{i=0}^{N-1} {a_i^0 \vw_i}$, where $\Vm$ is defined in~\eqref{eq:eigendecomposition}, $\e^0$ be the outliers that we are trying to detect,  represented as $\e^0 = \sum_{i \in \Eset} {b_i \ImP_i}$, where $\ImP_i$ is impulse on the $i$th node, and $\Eset$ contains the outlier indices, that is, $\Eset \subset \{0, 1, 2, \cdots N-1\}  $, and $ \t = \x^0 + \e^0$ be the measurement.   
\begin{myLem}
\label{lem:s}
Let $\xw,\ew$ be the solution of ~\eqref{eq:AD_Opt}, and let $\xw = \Vm\aw = \sum_{i=0}^{N-1} { \widehat{a}_i \vw_i}$. Then, 
$$\ew = \Vm (\a^0 -\aw) + \sum_{i \in \Eset}   {b_i \ImP _i}. $$
\end{myLem}
\begin{proof}
\begin{displaymath}
  \ew \ \stackrel{(a)}{=}  \ \t - \xw
  \ \stackrel{(b)}{=}  \ \x^0 + \e^0 - \xw
  \ \stackrel{(c)}{=} \ \Vm \a^0 + \sum_{i \in \Eset} {b_i \ImP_i} - \Vm \aw,
\end{displaymath}
where $(a)$ follows from the feasibility of  $\xw,\ew$ in~\eqref{eq:AD_Opt_cond_2}; $(b)$ from the definition of $\t$; and $(c)$ from the definitions of  $\x^0$ and $\xw$.
\end{proof}
Lemma~\ref{lem:s} provides another representation of the outliers, which is useful in the following theorem.

Let $\K = (\Id - \Lambda)^T \Vm^T \Vm (\Id - \Lambda)$, the $\K$ norm as  $\left\|\x\right\|_\K = \sqrt{\x^T \K \x}$, and $\Kset_\eta = \{ \a \in \R^N: \left\|\a\right\|_{\K} \leq \eta,  {\rm for~all~} \a \neq 0 \}$.
\begin{myLem}
\label{lem:tv}
Let ~$\x =  \Vm \a$ satisfy~\eqref{eq:rAD_Opt_cond1},~\eqref{eq:rAD_Opt_cond2}, and $\a \neq 0 $. Then, $ \a \in \Kset_\eta$.
\end{myLem}
\begin{proof}
\begin{eqnarray*}
  \STV_{2}(\x) & = & \left\| \x - \Adj \x \right\|_2^2
  \ {=}  \ \left\| \Vm \a -  \Adj \Vm \a \right\|_2^2,
  \\ \nonumber
  & \stackrel{(a)}{=}  & \left\| \Vm \a -  \Vm \Lambda \a \right\|_2^2
  \ {=}  \ \a^T (\Id - \Lambda)^T \Vm^T \Vm (\Id - \Lambda) \a,
  \\ \nonumber
  & \stackrel{(b)}{=}  & \a^T \K \a  \ \stackrel{(c)}{\leq} \  \eta^2,
\end{eqnarray*}
where $(a)$ follows from~\eqref{eq:eigendecomposition}; $(b)$ from the definition of the $\K$ norm; and $(c)$ from the feasibility of $\x$.
\end{proof}
Lemma~\ref{lem:tv} shows that the frequency components of the solution from~\eqref{eq:AD_Opt} and~\eqref{eq:rAD_Opt} belong to a subspace, which is useful in the following theorem.

\begin{myThm}
\label{thm:ad}
Let $ \STV_{2}(\x^0) \leq \eta^2$, $\x^0 \neq 0$, $ \left\|\e^0\right\|_0 \leq k $, and $\xw, \ew$ be the solution of~\eqref{eq:AD_Opt} with $\xw \neq 0$. Let $\Kset_{2\eta}$ have the following property: 
$$
\left\|\Vm\a\right\|_0 \geq 2k+1~~~~~{\rm for~all~} \a \in \Kset_{2\eta},
$$
where $k \geq \left\|\e^0\right\|_0$. Then,  perfect recovery is achieved,
$$
\xw  \ = \ \x^0, \ew  \ = \ \e^0.
$$
\end{myThm}
\begin{proof}
Since both $ \x^0 = \Vm \a^0$ and $\xw = \Vm \aw$ are feasible solutions of ~\eqref{eq:AD_Opt}, by Lemma~\ref{lem:tv}, we then have that $\a^0, \aw \in \Kset_\eta$. We next bound their difference, $\a^0 - \aw $, by the triangle inequality, as  $\left\|\a^0 - \aw \right\|_{\K} \leq \left\|\a^0\right\|_{\K} + \left\| \aw \right\|_{\K} \leq 2 \eta$.

If $\a^0 \neq \aw$, then $\a^0-\aw \in \Kset_{2\eta}$, so that $\left\|\Vm(\a^0-\aw) \right\|_0 \geq 2k+1$. From Lemma~\ref{lem:s},  we have $\left\|\ew\right\|_0 = \left\|\Vm (\a^0-\aw) + \sum_{i \in \Eset} {b_i \e_i} \right\|_0 \geq k+1$. The last inequality comes from the fact that at most $k$ indices can be canceled by the summation.

On the other hand, $\ew$ is the optimum of~\eqref{eq:AD_Opt}, thus, $\left\|\ew\right\|_0 \leq \left\|\e^0\right\|_0 = k$, which leads to a contradiction.

Therefore, $\aw = \a^0$ and $\ew = \e^0$.
\end{proof}

Theorem~\ref{thm:ad} shows the condition that leads to perfect outlier detection  by following~\eqref{eq:AD_Opt}. The key is that the outliers are sufficiently sparse and the smooth graph signal is not sparse.

\begin{myThm}
\label{thm:gad}
Let $ \STV_{2}(\x^0) \leq \eta^2$, $\x^0 \neq 0$, $ \left\|\e^0\right\|_0 \leq k $, and $\xw, \ew$ be the solution of~\eqref{eq:rAD_Opt} with $\xw \neq 0$. Let $\Kset_{2\eta}$ have the following property: 
$$
\left\|\left( \Vm\a \right)_{\Eset^c}\right\|_1 > \left\|\left( \Vm \a \right)_{\Eset}\right\|_1~~~~~{\rm for~all~} \a \in  \Kset_{2\eta}
$$
where $\Eset^c \cap \Eset$ is the empty set, $\Eset^c \cup \Eset = \{0, 1, 2, \cdots N-1 \}$. Then,  perfect recovery is achieved,
$$
\xw  \ = \ \x^0, \ew  \ = \ \e^0.
$$
\end{myThm}

\begin{proof}
From Lemma~\ref{lem:s}, we have
\begin{eqnarray}
\left\|\ew\right\|_1 & = & \left\|\Vm(\a^0-\aw) + \sum_{i \in \Eset} {b_i \ImP_i} \right\|_1
\nonumber \\ \nonumber
& = & \left\| \left( \Vm(\a^0-\aw) \right)_{\Eset^c} + \left( \Vm (\a^0-\aw) \right)_{\Eset} + \sum_{i \in \Eset} {b_i \ImP_i} \right\|_1
\\ \nonumber
& = & \left\|\left( \Vm (\a^0-\aw) \right)_{\Eset^c} \right\|_1 + \left\| \left( \Vm (\a^0-\aw) \right)_{\Eset} + \sum_{i \in \Eset} {b_i \ImP_i} \right\|_1
\end{eqnarray}

Denote $\left( \Vm  (\a^0-\aw) \right) _{\Eset}  =  \sum_{i \in \Eset} {d_i \ImP_i} $, we further have
\begin{eqnarray}
\left\|\ew\right\|_1 & = & \left\|\left(  \Vm (\a^0-\aw) \right)_{\Eset^c} \right\|_1 + \left\|\sum_{i \in \Eset} {(d_i+b_i) \ImP_i} \right\|_1
\nonumber \\ \nonumber
& = & \left\|\left( \Vm  (\a^0-\aw) \right)_{\Eset^c} \right\|_1 + \sum_{i \in \Eset} |d_i+b_i|
\end{eqnarray}
If $\a^0 \neq \aw$, then $\a^0-\aw \in \Kset_{2\eta}$.  By the
assumption, we have $ \left\| \left( \Vm (\a^0-\aw) \right)_{\Eset^c}
\right\|_1 > \left\| \left( \Vm (\a^0-\aw) \right)_{\Eset} \right\|_1
= \left\|\sum_{i \in \Eset} {d_i \ImP_i} \right\|_1 = \sum_{i \in
  \Eset} |d_i|$. We thus obtain
\begin{eqnarray*}
  \left\|\ew\right\|_1 & = & \left\|\left( \Vm (\a^0-\aw) \right)_{\Eset^c} \right\|_1 + \sum_{i \in \Eset} |d_i+b_i|
  \nonumber \\ \nonumber
  & > & \sum_{i \in \Eset} ( |d_i| + |d_i+b_i| )
  \ \geq \ \sum_{i \in \Eset} |b_i|.
\end{eqnarray*}
On the other hand, $\ew$ is the optimum of~\eqref{eq:rAD_Opt}, so $\left\|\ew\right\|_1 \leq \left\|\e^0\right\|_1 = \left\|\sum_{i \in \Eset} {b_i \ImP_i} \right\|_1 = \sum_{i \in \Eset} |b_i|$, which leads to a contradiction.

Therefore, $\aw = \a^0$ and $\ew = \e^0$.
\end{proof}
Theorems~\ref{thm:ad} and~\ref{thm:gad} show that under appropriate assumptions,~\eqref{eq:AD_Opt}, ~\eqref{eq:rAD_Opt} detects the outliers perfectly. Note that the assumptions on $\Kset$ in Theorems~\ref{thm:ad} and~\ref{thm:gad} are related to two factors: the upper bound on smoothness, $\eta$, and the eigenvector matrix, $\Vm$. The volume of $\Kset_{2\eta}$ is determined by the upper bound on smoothness, $\eta$. The mapping properties of $\Vm$ are also restricted by Theorems~\ref{thm:ad} and~\ref{thm:gad}. For instance, in Theorem~\ref{thm:ad}, the eigenvector matrix should map each element in $\Kset_{2\eta}$ to a non-sparse vector.

\mypar{Robust graph signal inpainting}
One problem with the graph signal inpainting in Section~\ref{sec:gsi} is that it tends to trust the accessible part, which may contain sparse, but large-magnitude outliers. Robust graph signal inpainting should prevent the solution from being influenced by the outliers. We thus consider the following optimization problem:
  \begin{eqnarray}
  \label{eq:RGSI_Opt}
   	\xw, \ww, \ew &=& \arg \min_{\x, \w, \e}\ \alpha \STV_2(\x)  + \gamma \left\|\e\right\|_0,
   	\\ 
   \label{eq:RGSI_Opt_cond1}
	\text{subject to}
	&& \left\|\w \right\|_F^2 < \eta^2
	  \\ 
	  \label{eq:RGSI_Opt_cond2}
	&& \t_\M = (\x + \w + \e )_\M;
   \end{eqnarray}
   this is a special case of~\eqref{eq:Opt} with $L = 1, \beta = 0$;
   see Table~\ref{table:algorithms}.

   Similarly to~\eqref{eq:C_Opt}, instead of dealing with the $\ell_0$
   norm, we relax it to be the $\ell_1$ norm and
   reformulate~\eqref{eq:RGSI_Opt} as an unconstrained problem,
  \begin{eqnarray}
  \label{eq:RRGSI_Opt_unconstraint}
   	\xw, \ew  \ =\ \arg \min_{\x, \e}\,  \left\| \t_{\M} - (\x+\e)_{\M} \right\|^2 + \alpha \STV_{2}(\x)  + \gamma \left\|\e\right\|_1.
   \end{eqnarray}
We call problem~\eqref{eq:RRGSI_Opt_unconstraint} the~\emph{robust graph total variation regularization} (RGTVR) problem. In~\eqref{eq:RRGSI_Opt_unconstraint}, we merge conditions~\eqref{eq:RGSI_Opt_cond1} and~\eqref{eq:RGSI_Opt_cond2} and move them from the constraint to the objective function. Note that~\eqref{eq:RRGSI_Opt_unconstraint} combines anomaly detection and graph signal inpainting to provide a twofold inpainting. The first level detects the outliers in the accessible part and provides a clean version of the accessible measurement; the second level uses the clean measurement  to recover the inaccessible part.  We solve~\eqref{eq:RRGSI_Opt_unconstraint} by using ADMM, as discussed in Section~\ref{sec:gsr}. For implementation details, see Algorithm~\ref{alg:RGTVR}.

\begin{algorithm}[t]
  \footnotesize
  \caption{\label{alg:RGTVR} Robust Graph Total Variation Regularization}
  \begin{tabular}{@{}ll@{}}
    \addlinespace[1mm]
    {\bf Input}
    &  $\t$~~~input graph signal \\
    {\bf Output}
    &  $\ew$~~~outlier graph signal \\
    &  $\xw$~~~output graph signal  \\
    \addlinespace[2mm]
    {\bf Function} & {\bf RGTVR($\t$)} \\
    &while the stopping criterion is not satisfied  \\
    & ~~$\x \leftarrow \left(\Id+2\alpha\eta^{-1} \widetilde{\Adj}  \right)^{-1}(\t-\e-\w-\c- \eta^{-1}\lambda)$  \\
    &~~$\w \leftarrow \eta{(\t-\x-\w-\c- \eta^{-1}\lambda)} / (\eta+2)$  \\
    &~~$\e \leftarrow \Shrink_{\gamma\eta^{-1}}{(\t-\x-\e-\c- \eta^{-1}\lambda)}$  \\
    &~~$\lambda \leftarrow \lambda - \eta (\t-\x-\e-\w-\c)$  \\
    &~~$\c_\M \leftarrow 0, \c_\U \leftarrow (\t-\x-\w-\e- \eta^{-1}{\lambda})_\U$, \\
    &end \\
    &{\bf return} $\ew \leftarrow \e,  ~~\xw \leftarrow \x$ \\
    \addlinespace[1mm]
  \end{tabular}
\end{algorithm}

\section{Experimental Results}
\label{sec:results}
We now evaluate the proposed methods on several real-world recovery
problems. Further, we apply graph signal inpainting and robust graph signal
inpainting to online blog classification and bridge condition
identification for indirect bridge structural health monitoring; We
apply graph signal matrix completion to temperature estimation and
expert opinion combination.

\mypar{Datasets}
We use the following datasets in the experiments:

\setcounter{subsubsection}{0}
\subsubsection{Online blogs}
\label{sec:blogs }
We consider the problem of classifying $N=1224$ online political blogs
as either conservative or liberal~\cite{AdamicG:05}. We represent
conservative labels as $+1$ and liberal ones as $-1$.  The blogs are
represented by a graph in which nodes represent blogs, and directed
graph edges correspond to hyperlink references between blogs. For a
node $v_n$, its outgoing edges have weights $1/\deg(v_n)$, where
$\deg(v_n)$ is the out-degree of $v_n$ (the number of outgoing edges). The graph signal here is  the label assigned to the blogs.

\subsubsection{Acceleration signals}
\label{sec:accelerationsignals}
We next consider the bridge condition identification
problem~\cite{CerdaGBRBZCM:12,CerdaCBGRK:12}. To validate the
feasibility of indirect bridge structural health monitoring, a
lab-scale bridge-vehicle dynamic system was built. Accelerometers were
installed on a vehicle that travels across the bridge; acceleration
signals were then collected from those accelerometers. To simulate the severity of
different bridge conditions on a lab-scale bridge, masses with various
weights were put on the bridge. We collected 30 acceleration signals
for each of 31 mass levels from 0 to 150 grams in steps of 5 grams, to
simulate different degrees of damages, for a total of 930
acceleration signals. For more details on this dataset,
see~\cite{LedermanWBNGCKCR:14}.

The recordings are represented by an $8$-nearest neighbor graph, in
which nodes represent recordings, and each node is connected to eight
other nodes that represent the most similar recordings. The graph
signal here is the mass level over all the acceleration signals. The
graph shift $\Adj$ is formed as $\Adj_{i,j} =
\Pj_{i,j}/\sum_i{\Pj_{i,j}}$, with
\begin{displaymath}
  \Pj_{i,j} \ = \ \exp \left( \frac{-N^2\left\|\f_i - \f_j \right\|_2} {\sum_{i,j} \left\|\f_i - \f_j \right\|_2} \right),
\end{displaymath}
and $\f_i$ is a vector representation of the features of the $i$th
recording. Note that $\Pj$ is a symmetric matrix that represents an
undirected graph and the graph shift $\Adj$ is an asymmetric matrix that
represents a directed graph, which is allowed by the framework
of~\DSPG. From the empirical performance, we find that a
directed graph provides much better results than an undirected graph.

\subsubsection{Temperature data}
\label{sec:temperature}
We consider 150 weather stations in the United States that record
their local temperatures~\cite{SandryhailaM:13}. Each weather station has 365 days of
recordings (one recording per day), for a total of 54,750 measurements.  The graph
representing these weather stations is obtained by measuring the
geodesic distance between each pair of weather stations. The nodes are
represented by an $8$-nearest neighbor graph, in which nodes represent
weather stations, and each node is connected to eight other nodes that
represent the eight closest weather stations. The graph signals here
are the temperature values recorded in each weather station.

The graph shift $\Adj$ is formed as $\Adj_{i,j} = \Pj_{i,j}/\sum_i{\Pj_{i,j}}$, with
\begin{displaymath}
  \Pj_{i,j} = \exp \left( -\frac{N^2 d_{i,j}}{ \sum_{i,j}{d_{i,j}} } \right),
\end{displaymath}
where $d_{i,j}$ is the geodesic distance between the $i$th and the
$j$th weather stations. Similarly to the acceleration signals, we normalize
$\Pj$ to obtain a asymmetric graph shift, which represents a directed
graph, to achieve better empirical performance.

\subsubsection{Jester dataset 1}
\label{sec:jester}
The Jester joke data set~\cite{Jester} contains $4.1 \times 10^6$
ratings of 100 jokes from 73,421 users. The graph representing the
users is obtained by measuring the $\ell_1$ norm of
existing ratings between each pair of jokes. The nodes are represented
by an 8-nearest neighbor graph in which nodes represent users and each
node is connected to eight other nodes that represent similar
users. The graph signals are the ratings of each user. The graph shift
$\Adj$ is formed as
\begin{displaymath}
  \Pj_{i,j} \ = \ \exp \left( \frac{-N^2\left\|\f_i - \f_j \right\|_1} {\sum_{i,j} \left\|\f_i - \f_j \right\|_1 } \right),
\end{displaymath}
where $\f_{i}$ is the vector representation of the existing ratings
for the $i$th user. Similarly to acceleration signals, we normalize
$\Pj$ to obtain an asymmetric graph shift, which represents a directed
graph, to achieve better empirical performance.

\mypar{Evaluation score}
To evaluate the performance of the algorithms, we use the following
four metrics: accuracy (ACC), mean square error (MSE), root mean
square error (RMSE), and mean absolute error (MAE), defined as
\begin{eqnarray*}
  {\rm ACC} & \ = \ & \frac{1}{N}  \sum_{i=1}^{N} {\bf 1}{( x_i = \widehat{x}_i )},\\
  {\rm MSE} & \ = \ & \frac{1}{N} \sum_{i=1}^{N} {(x_i - \widehat{x}_i)^2},\\
  {\rm RMSE} & \ = \ & \sqrt{  \frac{1}{N} \sum_{i=1}^{N} {(x_i - \widehat{x}_i)^2}  } = \sqrt{{\rm MSE} },\\
  {\rm MAE} & \ = \ & \frac{\sum_{i=1}^{N} {|x_i - \widehat{x}_i|}}{N},
\end{eqnarray*}
where $x_i $ is the ground-truth for the $i$th sample, $\widehat{x}_i $ is
the estimate for the $i$th sample, and ${\bf 1}$ is the indicator
function, ${\bf 1}(x) = 1$, for $x=0$, and $0$ otherwise.





In the following applications, the tuning parameters for each algorithm are chosen by cross-validation; that is, we split the accessible part into a training part and a validation part. We train the model with the training part and choose the tuning parameter that provides the best performance in the validation~part.

\mypar{Applications of graph signal inpainting}
Parts of this subsection have appeared in~\cite{ChenSLWMRBGK:14}; we
include them here for completeness. We apply the proposed graph signal
inpainting algorithm to online blog classification and bridge
condition identification. We compare the proposed GTVR~\eqref{eq:GTVR}
with another regression model based on graphs, graph Laplacian
regularization regression
(LapR)~\cite{ZhuGL:03,ZhouS:04,BelkinNS:06}. As described in
Section~\ref{sec:intro}, the main difference between LapR and GTVR is
that a graph Laplacian matrix in LapR is restricted to be symmetric
and only represents an undirected graph; a graph shift in GTVR can be
either symmetric or asymmetric.

\setcounter{subsubsection}{0}
\subsubsection{Online blog classification}
We consider a semi-supervised classification problem, that is,
classification with few labeled data and a large amount of unlabeled
data~\cite{Zhu:05}. The task is to classify the unlabeled blogs. We
adopt the dataset of blogs as described in Section~\ref{sec:blogs
}. We randomly labeled 0.5\%, 1\%, 2\%, 5\%, and 10\% of blogs, called the~\emph{labeling ratio}. We then applied the graph signal inpainting algorithms to estimate the labels for the remaining blogs. Estimated labels were thresholded at zero, so that positive values were set to $+1$ and negative to $-1$.

Classification accuracies of GTVR and LapR were then averaged over 30
tests for each labeling ratio and are shown in
Figure~\ref{fig:blog_acc}. We see that GTVR achieves significantly
higher accuracy than LapR for low labeling ratios. The failure of LapR
at low labeling ratios is because an undirected graph fails to reveal
the true structure.

Figure~\ref{fig:blog_acc} also shows that the performance of GTVR saturates at around $95\%$. Many of the remaining errors are misclassification of blogs with many connections to the blogs from a different class, which violates the smoothness assumption underlying GTVR. Because of the same reason, the performance of a data-adaptive graph filter also saturates at around $95\%$~\cite{ChenSMK:13}. To improve on this performance may require using a more sophisticated classifier that we will pursue in future work.

\begin{figure}[t]
  \begin{center}
    \begin{tabular}{c} 
      \includegraphics[width= 0.6\columnwidth]{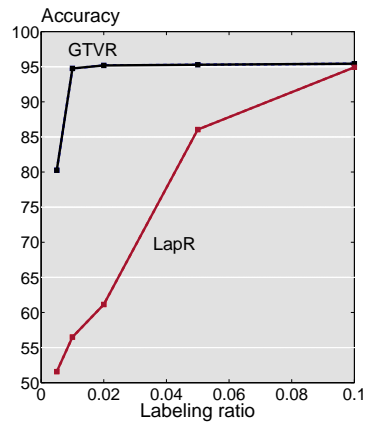}
    \end{tabular}
  \end{center}
  \vspace{-0.1in}
  \caption{\label{fig:blog_acc}   Accuracy comparison of online blog classification as a function of labeling ratio.}
\end{figure}

\subsubsection{Bridge condition identification}
We consider a semi-supervised regression problem, that is, regression
with few labeled data and a large amount of unlabeled
data~\cite{Zhu:05}. The task is to predict the mass levels of
unlabeled acceleration signals. We adopt the dataset of acceleration
signals as described in Section~\ref{sec:accelerationsignals}.  We
randomly assigned known masses to 0.5\%, 1\%, 2\%, 5\%, and 10\% of
acceleration signals and applied the graph signal inpainting
algorithms to estimate the masses for remaining nodes.

Figure~\ref{fig:bridge_mse} shows MSEs for estimated masses averaged
over 30 tests for each labeling ratio. The proposed GTVR approach
yields significantly smaller errors than LapR for low labeling
ratios. Similarly to the conclusion of online blog classification, a
direct graph adopted in GTVR reveals a better structure.

We observe that the performance of GTVR saturates at 3 in terms of MSE. This may be the result of how we obtain the graph. Here we construct the graph by using features from principal component analysis of the data. Since the data is collected with a real lab-scale model, which is complex and noisy, the principal component analysis may not extract all the useful information from the data, limiting the performance of the proposed method even with larger number of samples.

\begin{figure}[t]
  \begin{center}
    \begin{tabular}{c} 
      \includegraphics[width= 0.6\columnwidth]{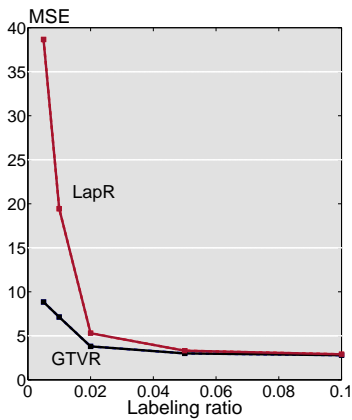}
    \end{tabular}
  \end{center}
  \vspace{-0.1in}
  \caption{\label{fig:bridge_mse}  MSE comparison for the bridge condition identification as a function of labeling ratio.}
\end{figure}

\mypar{Applications of graph signal matrix completion}
We now apply the proposed algorithm to temperature estimation,
recommender systems and expert opinion combination of online blog
classification. We compare the proposed GMCR~\eqref{eq:GMCR} with matrix completion algorithms. Those algorithms include
SoftImpute~\cite{MazumderHT:10}, OptSpace~\cite{KeshavanMO:10},
singular value thresholding (SVT)~\cite{CaiCS:10}, weighted
non-negative matrix factorization (NMF)~\cite{ZhangWM:06}, and
graph-based weighted nonnegative matrix factorization
(GWNMF)~\cite{GuZD:10}. Similarly to the matrix completion algorithm
described in Section~\ref{sec:matrixcompletion}, SoftImpute, OptSpace,
and SVT minimize the rank of a matrix in similar, but different
ways. NMF is based on matrix factorization, assuming that a matrix
can be factorized into two nonnegative, low-dimensional matrices;
GWNMF extends NMF by further constructing graphs on columns or rows to
represent the internal information. In contrast to the proposed
graph-based methods, GWNMF considers the graph structure in the hidden
layer. For a fair comparison, we use the same graph structure for GWNMF
and GMCM. NMF and GWNMF solve non-convex problems and get local
minimum.

\setcounter{subsubsection}{0}
\subsubsection{Temperature estimation} 
We consider matrix completion, that is, estimation of the
missing entries in a data matrix~\cite{CandesR:09}. The task is to
predict missing temperature values in an incomplete temperature data
matrix where each column corresponds to the temperature values of all
the weather stations from each day. We adopt the dataset of
temperature data described in Section~\ref{sec:temperature}~\cite{SandryhailaM:13}. In each day of temperature recording, we randomly hide $50\%, 60\%, 70\%,
80\%, 90\%$ measurements and apply the proposed matrix completion
methods to estimate the missing measurements. To further test the
recovery algorithms with different amount of data,
we randomly pick 50 out of 365 days of recording and conduct the same
experiment for 10 times. In this case, we have a graph signal matrix with $N = 150$, and $L = 50$, or $L = 365$.

Figures~\ref{fig:temperature_RMSE} and \ref{fig:temperature_MAE} show
RMSEs and MAEs for estimated temperature values averaged over 10
tests for each labeling ratio. We see that GTVM, as a pure graph-based method~\eqref{eq:GTVM}, performs well when the labeling ratio is low. When the labeling ratio increases, the performance of GTVM does not improve as much as the matrix completion algorithms, because it cannot learn from the graph signals. For both evaluation scores, RMSE and MAE, GMCR outperforms all matrix completion algorithms because it combines the prior information on graph structure with the low-rank assumption to perform a twofold learning scheme.

\begin{figure}[t]
  \begin{center}
    \begin{tabular}{cc}  
      \includegraphics[width=0.5\columnwidth]{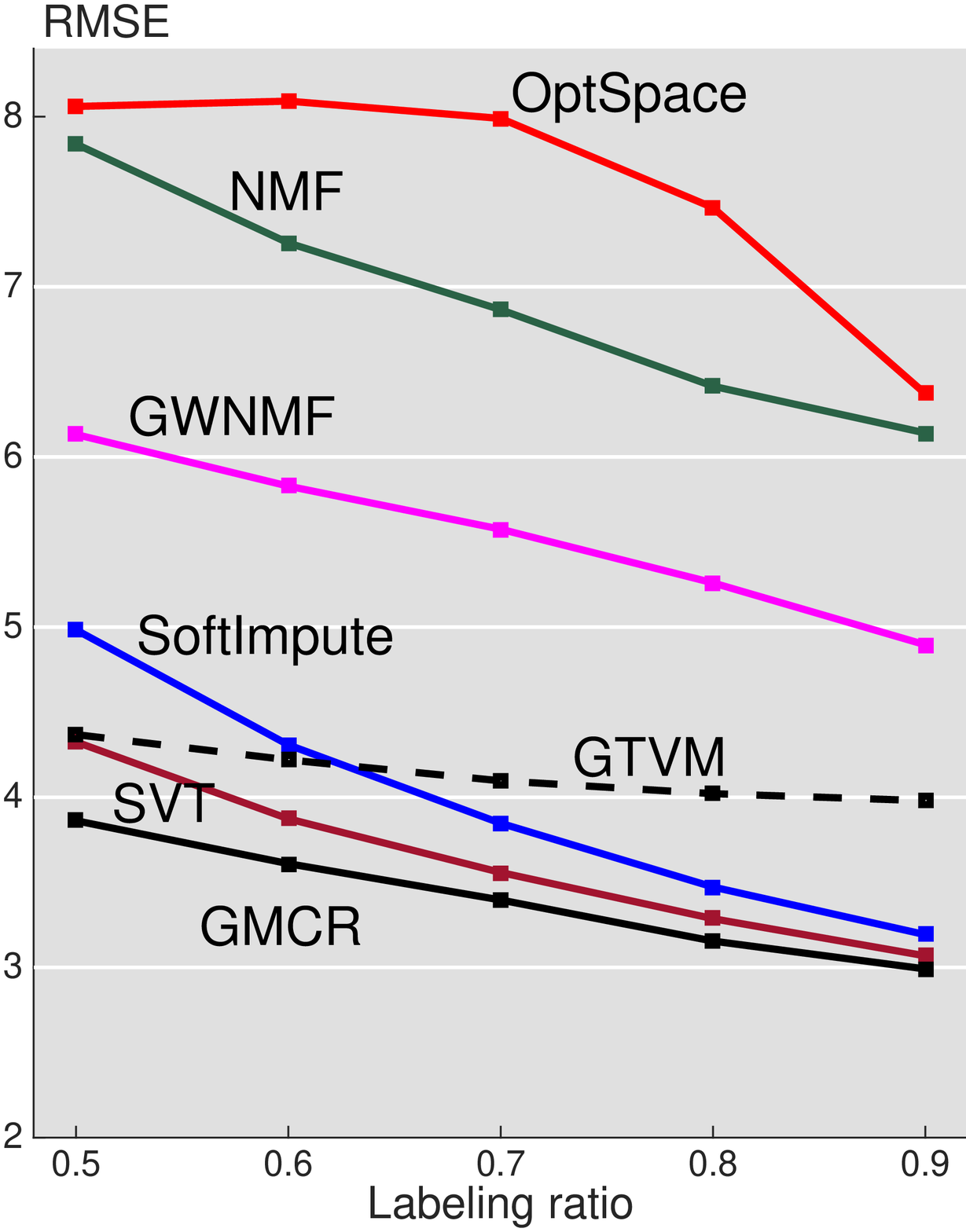} 
      & 
      \includegraphics[width=0.5\columnwidth]{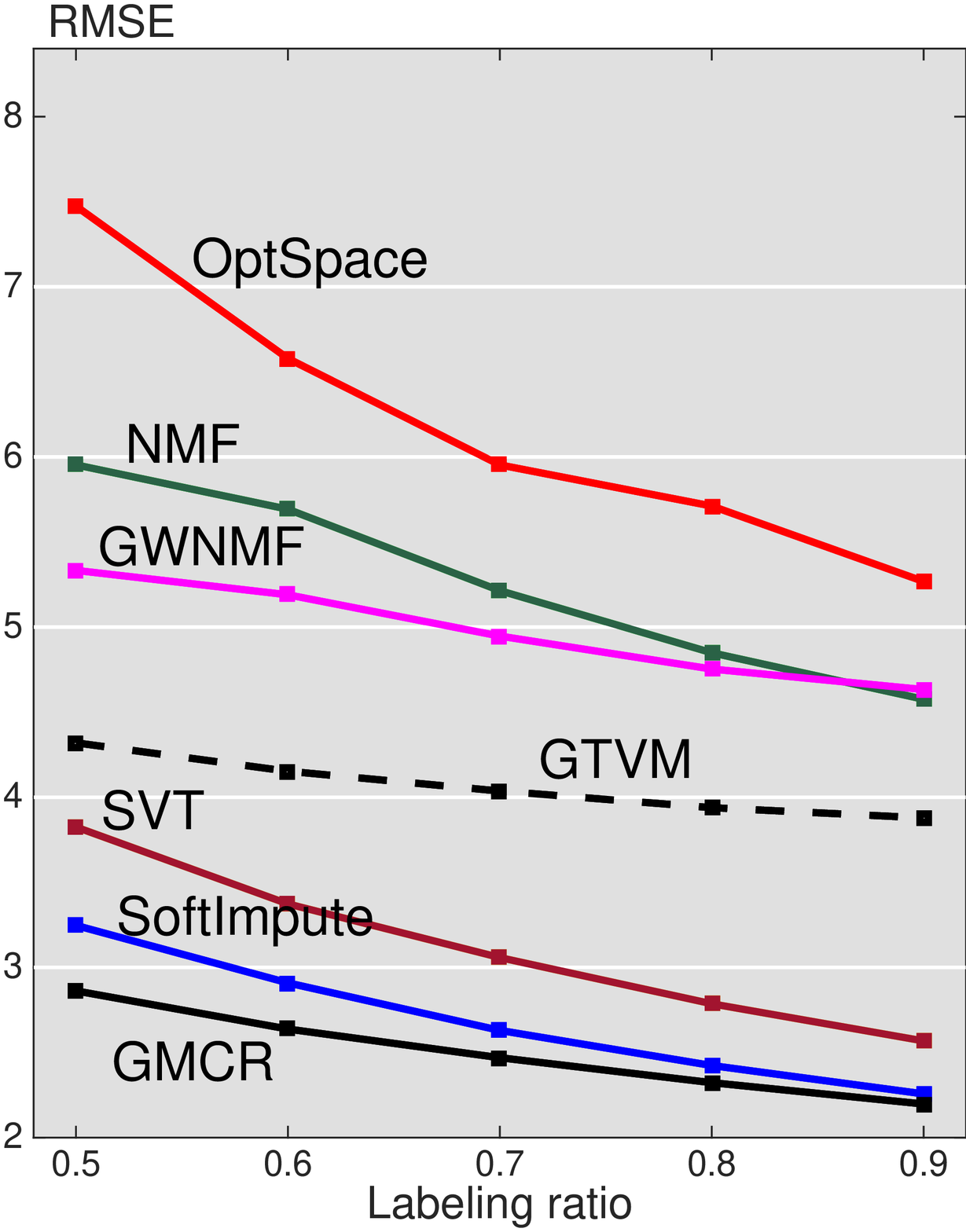}\\
      {\small (a) 50 recordings.} & {\small (b) 365 recordings.} 
    \end{tabular}
  \end{center}
  \vspace{-0.1in}
  \caption{\label{fig:temperature_RMSE} RMSE of temperature estimation for 50 recordings and 365 recordings.}
\end{figure}

\begin{figure}[t]
  \begin{center}
    \begin{tabular}{cc}  
      \includegraphics[width=0.48\columnwidth]{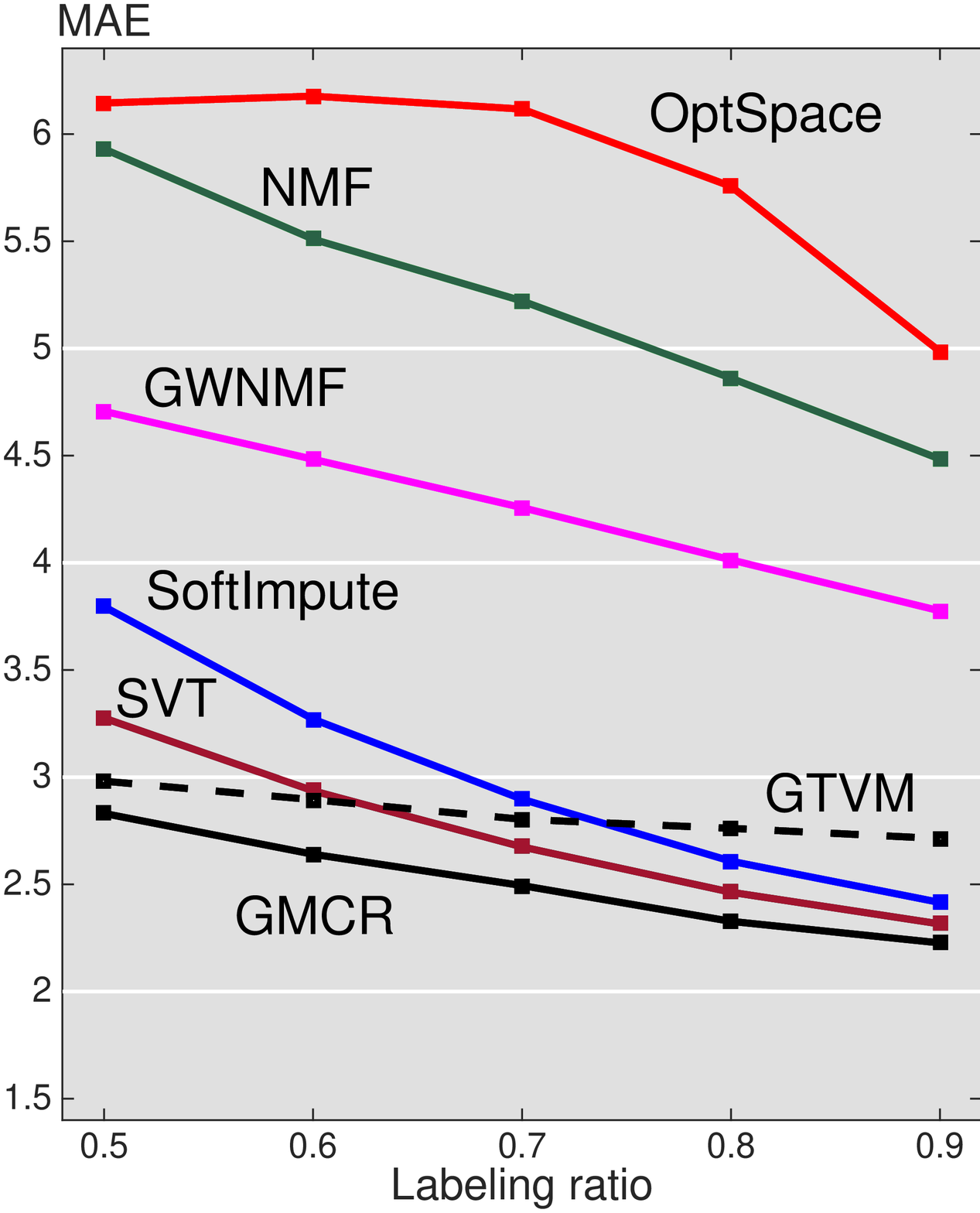} 
      & 
      \includegraphics[width=0.48\columnwidth]{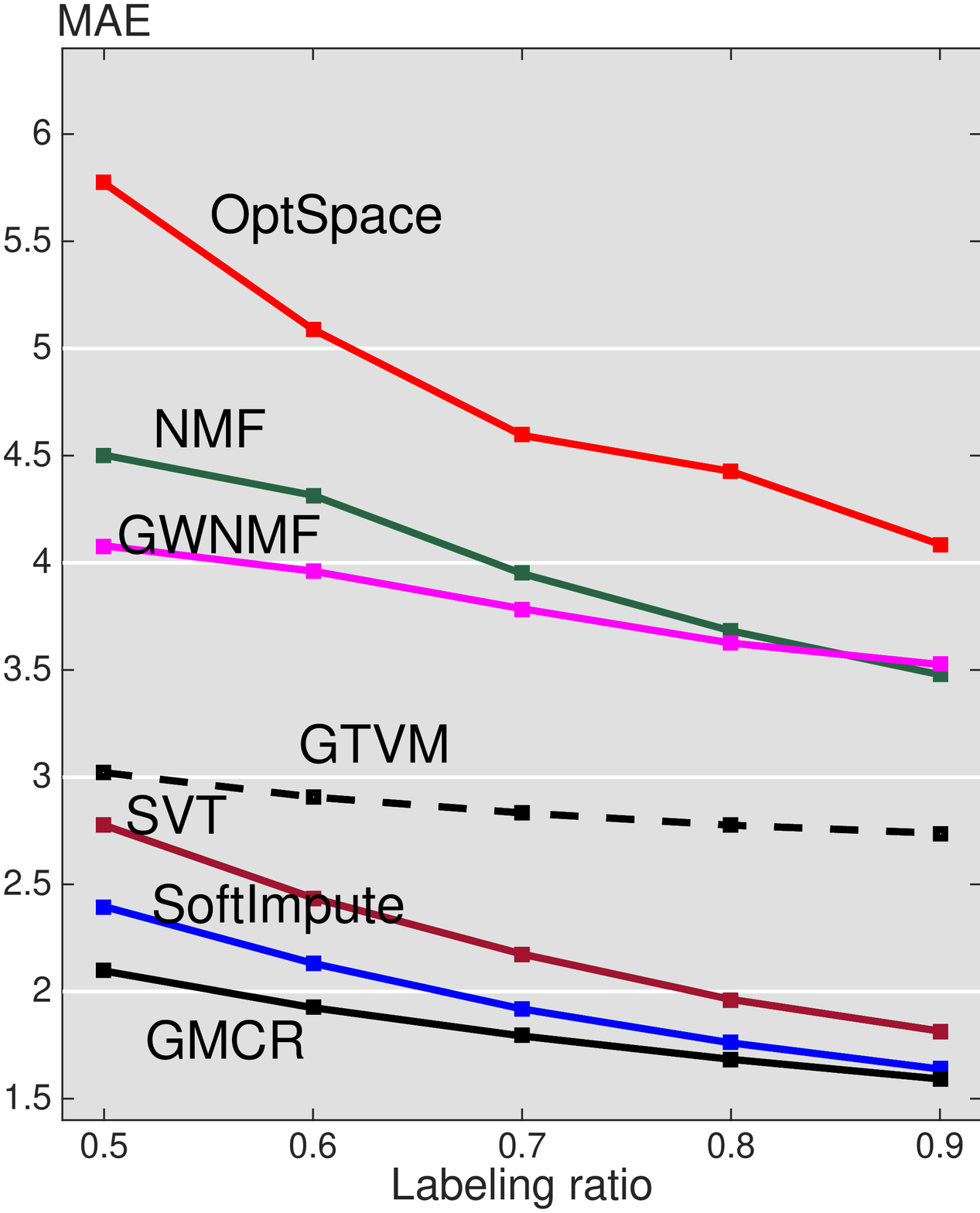}\\
      {\small (a) 50 recordings.} & {\small (b) 365 recordings.} 
    \end{tabular}
  \end{center}
  \vspace{-0.1in}
  \caption{\label{fig:temperature_MAE}  MAE of temperature estimation for 50 recordings and 365 recordings.}
\end{figure}

\subsubsection{Rating completion for recommender system}
We consider another matrix completion problem in the context of
recommender systems based on the Jester dataset 1~\cite{Jester}. The task is to predict
missing ratings in an incomplete user-joke rating matrix where each
column corresponds to the ratings of all the jokes from each
user. Since the number of users is large compared to the number of
jokes, following the protocol in~\cite{KeshavanMO:09}, we randomly
select 500 users for comparison purposes. For each user, we extract
two ratings at random as test data for 10 tests. In this case, we have a graph signal matrix with $N = 100$, and $L = 500$.

Figures~\ref{fig:jester_rmse} and \ref{fig:jester_mae} show RMSEs and
MAEs, defined in the evaluation score section, for estimated temperature values averaged over 10 tests.  We see
that graph-based methods (GWNMF and GMCR) take the advantage of
exploiting the internal information of users and achieve smaller
error. For RMSE, GMCR provides the best performance; for MAE, GWNMF
provides the best performance.

\begin{figure}[t]
  \begin{center}
    \begin{tabular}{c}  
      \includegraphics[width=0.9\columnwidth]{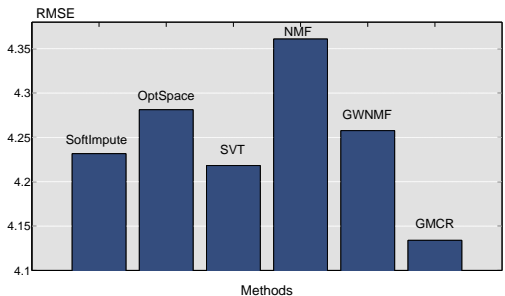} 
\\
    \end{tabular}
  \end{center}
  \vspace{-0.1in}
  \caption{\label{fig:jester_rmse}  RMSE of the rating completion in Jester 1 dataset.}
\end{figure}

\begin{figure}[t]
  \begin{center}
    \begin{tabular}{c}  
      \includegraphics[width=0.9\columnwidth]{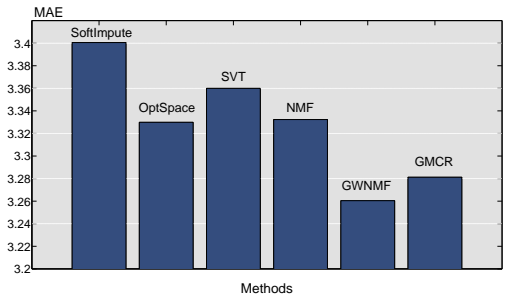} 
\\
    \end{tabular}
  \end{center}
  \vspace{-0.1in}
  \caption{\label{fig:jester_mae}  MAE of the rating completion in Jester~1 dataset.}
\end{figure}

\subsubsection{Combining expert opinions}
In many real-world classification problems, the opinion of experts determines the ground truth. At times, these are hard to obtain; for instance, when a dataset is too large, obtaining the opinion of experts is
too expensive, or experts differ among themselves, which
happens, for example, in biomedical image
classification~\cite{KuruvillaSHK:13}.  In this case, a popular
solution is to use multiple experts, or classifiers, to label
dataset elements and then combine their opinions into the final
estimate of the ground truth~\cite{CholletiGBPDSP:09}.  As we
demonstrate here, opinion combining can be formulated and solved as
graph signal matrix denoising.

We consider the online blog classification problem.  We hide the
ground truth and simulate $K=100$ experts labeling $1224$ blogs.
Each expert labels each blog as conservative ($+1$) or liberal ($-1$)
to produce an opinion vector $\t_k \in\{+1,-1\}^{1224}$. Note that
labeling mistakes are possible. We combine opinions from all the
experts and form an opinion matrix $\T \in\{+1,-1\}^{1224 \times
  100}$, whose $k$th column is $\t_k$.  We think of $\T$ as a graph
signal matrix with noise that represents the experts' errors. We
assume some blogs are harder to classify than others (for instance,
the content in a blog is ambiguous, which is hard to
label), we split the dataset of all the blogs into ``easy'' and
``hard'' blogs and assume that there is a $90\%$ chance that an expert
classifies an ``easy'' blog correctly and only a $30\%$ chance that an
expert classifies a ``hard'' blog correctly.  We consider four cases
of ``easy'' blogs making up $55\%$, $65\%$, $75\%$, and $85\%$ of the
entire dataset.

A baseline solution is to average (AVG) all the experts opinions into
vector $\t_{\rm avg} = (\sum_k \t_k)/K$ and then use the sign
$\rm{sign}(\t_{\rm avg})$ vector as the labels to blogs.  We compare
the baseline solution with the GTVR solution~\eqref{eq:solution} and GMCR. In GTVR, we first denoise every signal $\t_k$ and then compute the
average of denoised signals $\widetilde{\t}_{\rm avg}=(\sum_k
\widetilde{\t}_k)/K$ and use ${\rm sign}(\widetilde{\t}_{\rm avg})$ as labels to blogs.

Using the proposed methods, we obtain a denoised opinion matrix. We
average the opinions from all the experts into a vector and use its
signs as the labels to blogs. Note that, for GTVR and GMCR, the
accessible part is all the indices in the opinion matrix $\T$; since
each entry in $\T$ can be wrong, no ground-truth is available for
cross-validation. We vary the tuning parameter and report the best
results. Figure~\ref{fig:blog_expert} shows the accuracy of estimating
the ground-truth. We see that, through promoting the smoothness in each
column, GTVR improves the accuracy; GMCR provides the best results
because of its twofold learning scheme. Note that Figure~\ref{fig:blog_expert}  does not show that the common matrix
completion algorithms provide the same ``denoised'' results as the
baseline solution.

\begin{figure}[t]
  \begin{center}
    \begin{tabular}{c} 
      \includegraphics[width= 0.48\columnwidth]{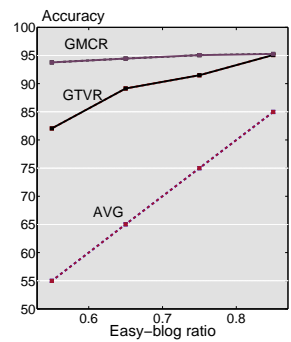}
    \end{tabular}
  \end{center}
  \vspace{-0.1in}
  \caption{\label{fig:blog_expert}   Accuracy of combining expert opinions.}
\end{figure}

\mypar{Applications of robust graph signal inpainting}
We now apply the proposed robust graph signal inpainting algorithm to
online blog classification and  bridge condition
identification. In contrast to what is done in the applications of graph signal inpainting, we manually add some outliers to the accessible part and compare the algorithm to
common graph signal inpainting algorithms.

\setcounter{subsubsection}{0}
\subsubsection{Online blog classification} 
We consider semi-supervised online blog classification as described in
Section~\ref{sec:blogs }. To validate the robustness of detecting
outliers, we randomly mislabel a fraction of the labeled blogs, feed them
into the classifiers together with correctly labeled signals, and
compare the fault tolerances of the
algorithms. Figure~\ref{fig:robust_blog_acc} shows the classification
accuracies when $1\%, 2\%$, and $5\% $ of blogs are labeled, with
$16.66\%$ and $33.33\%$ of these labeled blogs mislabeled in each
labeling ratio. We see that, in each case, RGTVR provides the most
accurate classification.

\begin{figure}[t]
  \begin{center}
    \begin{tabular}{cc}  
      \includegraphics[width=0.48\columnwidth]{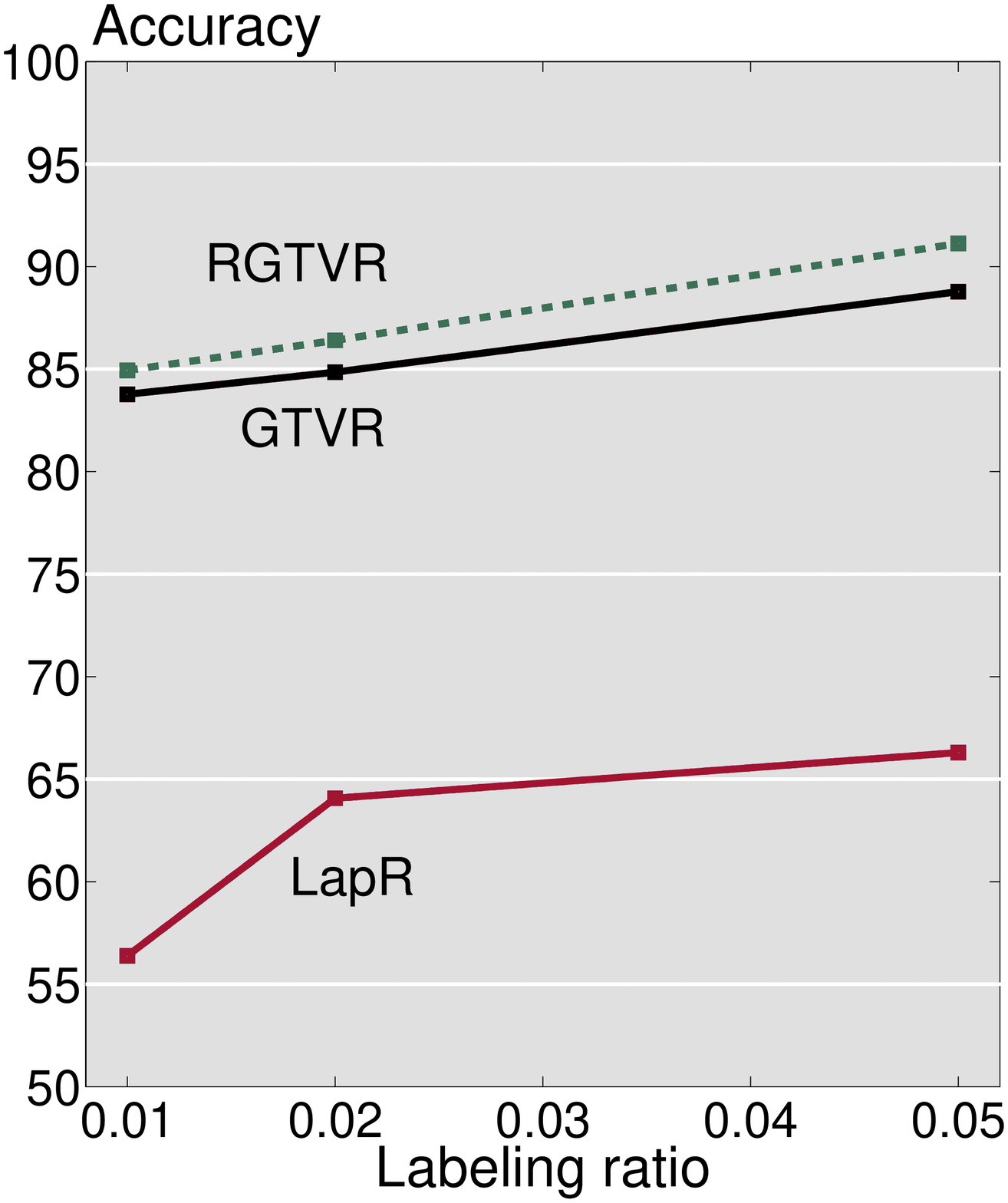} 
      & 
      \includegraphics[width=0.48\columnwidth]{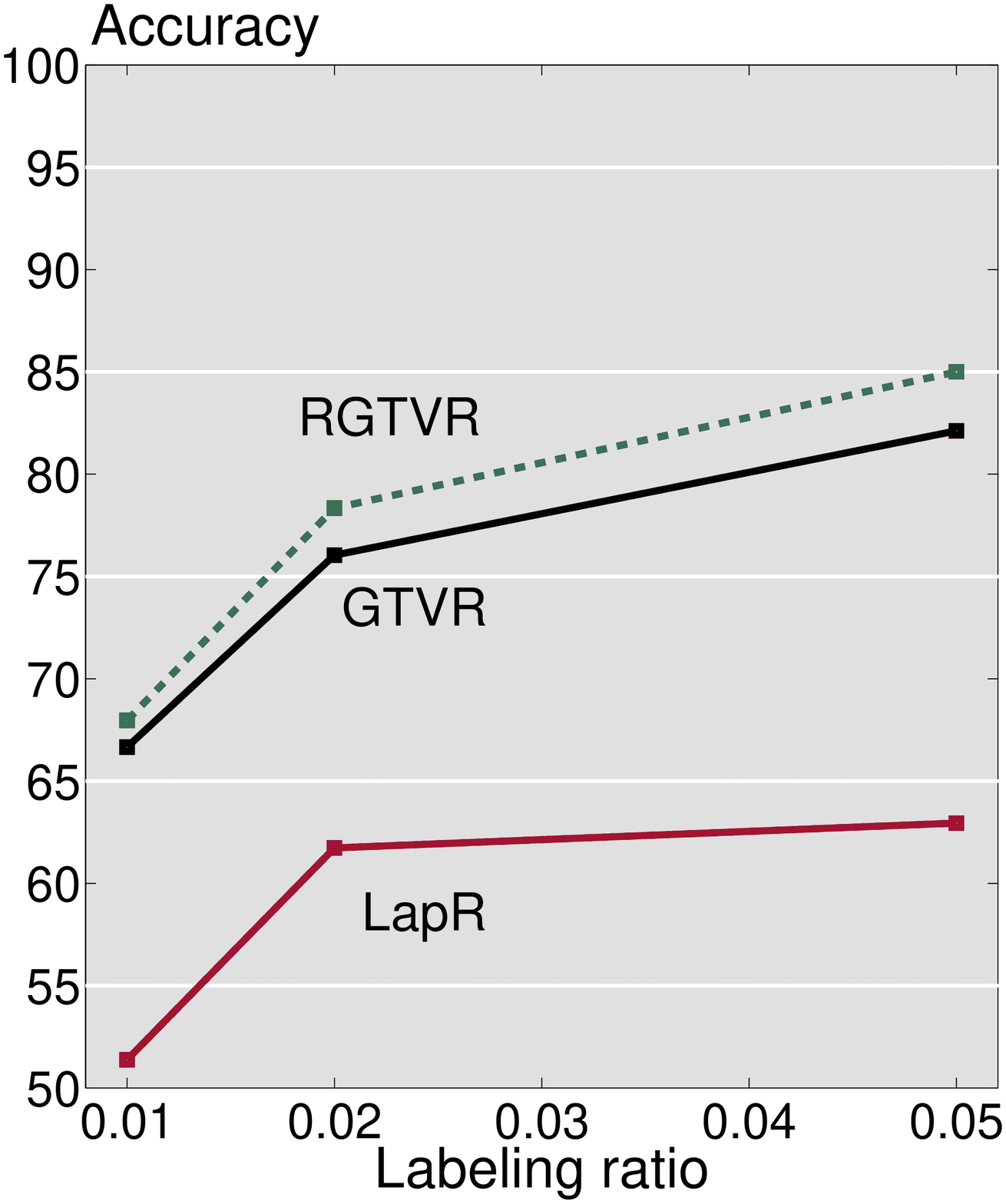}\\
      {\small (a) mislabeled ratio: 16.67\%.} & {\small (b) mislabeled ratio: 33.33\%.} 
    \end{tabular}
  \end{center}
  \vspace{-0.1in}
  \caption{\label{fig:robust_blog_acc} Robustness to mislabeled blogs: accuracy comparison with labeling ratio of $1\%, 2\%$ and $5\% $, and mislabeling ratio of $16.66\%$ and $33.33\%$ in each labeling ratio.}
\end{figure}

\subsubsection{Bridge condition identification}
We consider a semi-supervised regression problem and adopt the dataset
of acceleration signals as described in
Section~\ref{sec:accelerationsignals}. To validate the robustness of
facing outliers, we randomly mislabel a fraction of labeled
acceleration signals, feed them into the graph signal inpainting
algorithm together with correctly labeled acceleration signals, and
compare the fault tolerances of the
algorithms. Figure~\ref{fig:robust_bridge_mse} shows MSEs when $1\%,
2\%$, and $5\% $ of signals are labeled, with $16.66\%$ and $33.33\%$
of these labeled signals mislabeled in each labeling ratio. We see
that, in each case, RGTVR provides the smallest error.

\begin{figure}[t]
  \begin{center}
    \begin{tabular}{cc}  
      \includegraphics[width=0.48\columnwidth]{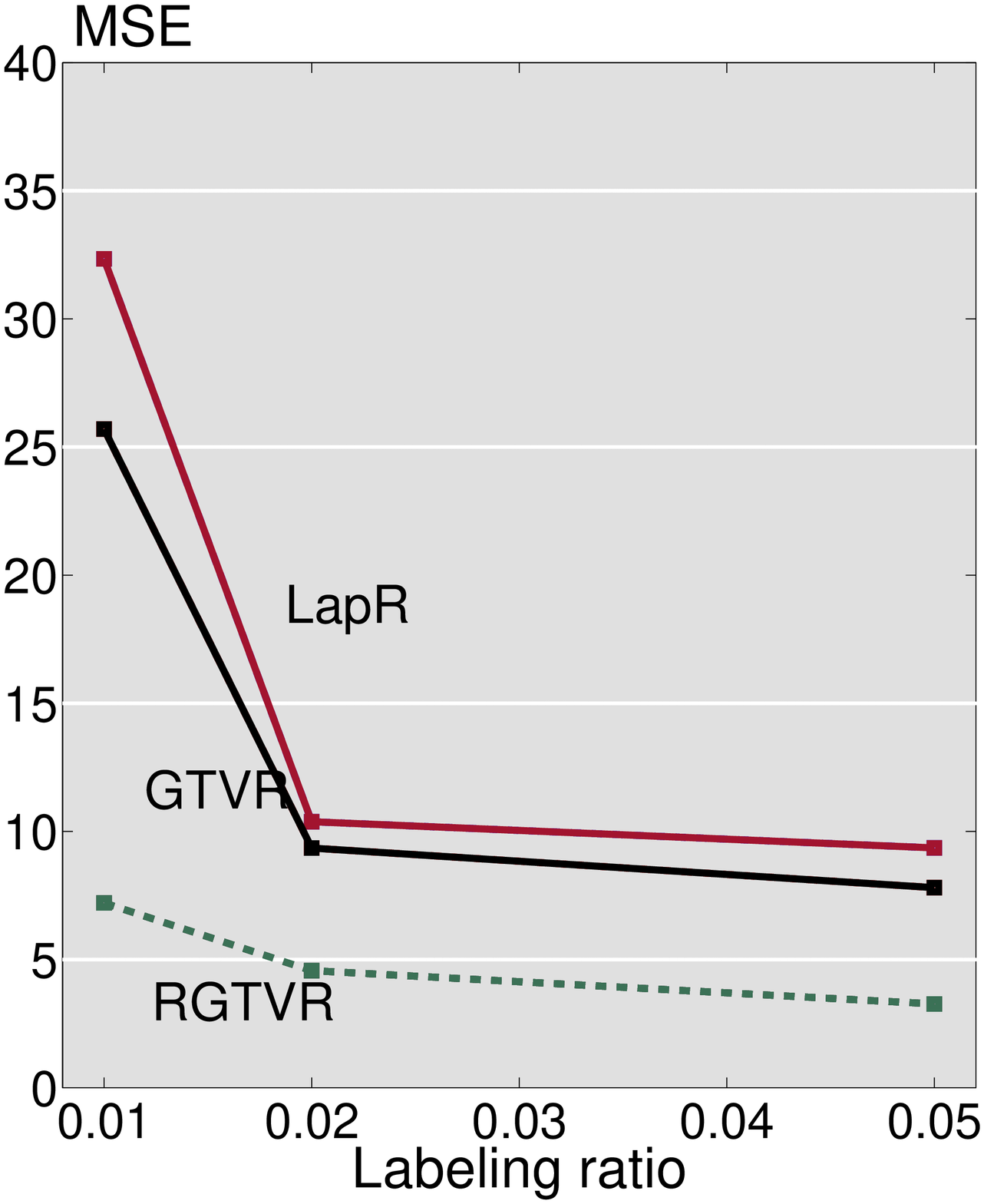} 
      & 
      \includegraphics[width=0.48\columnwidth]{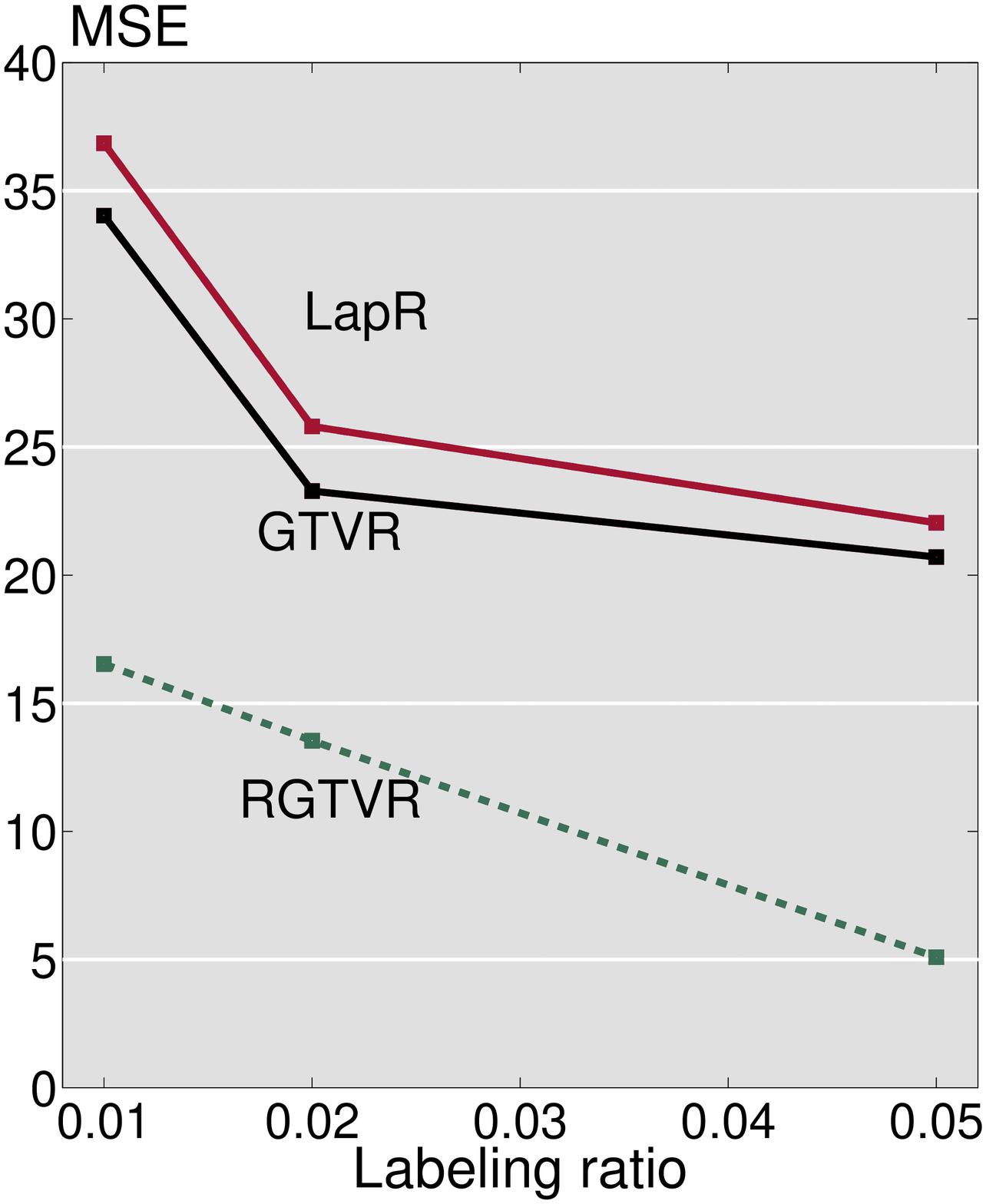}\\
      {\small (a) 16.67\%.} & {\small (b) 33.33\%.} 
    \end{tabular}
  \end{center}
  \vspace{-0.1in}
  \caption{\label{fig:robust_bridge_mse} Robustness to mislabeled signals: MSE
 comparison with labeling ratio of $1\%, 2\%$, and $5\% $, and mislabeling ratio of $16.66\%$ and $33.33\%$ in each labeling ratio.}
\end{figure}

\section{Conclusions}
\label{sec:conclusions}
We formulated graph signal recovery as an optimization problem and
provided a general solution by using the alternating direction method
of multipliers. We showed that several existing recovery
problems, including signal inpainting, matrix completion, and robust
principal component analysis, are related to the proposed graph signal
recovery problem. We further considered three subproblems, including
graph signal inpainting, graph signal matrix completion, and anomaly
detection of graph signals. For each subproblem, we provided specific
solutions and theoretical analysis. Finally, we validated the proposed
methods  on real-world recovery problems,
including online blog classification, bridge condition identification,
temperature estimation, recommender system, and expert opinion
combination of online blog classification.

\section{Acknowledgment}
We gratefully acknowledge support from the NSF through awards 1130616, 1421919, 1011903, 1018509, AFOSR grant FA95501210087, and the University Transportation Center grant-DTRT12-GUTC11
from the US Department of Transportation as well as the CMU Carnegie Institute of Technology Infrastructure Award. We also thank the editor and the reviewers for comments that led to improvements in the manuscript. We follow the principles of reproducible research. To that end, we created a reproducible research page available to
readers~\cite{ChenSMK:14:web}. Initial parts of this work were presented at ICASSP 2014~\cite{ChenSLWMRBGK:14}.

\bibliographystyle{IEEEtran}
\bibliography{bibl_jelena}

\begin{thebibliography}{10}
\providecommand{\url}[1]{#1}
\csname url@samestyle\endcsname
\providecommand{\newblock}{\relax}
\providecommand{\bibinfo}[2]{#2}
\providecommand{\BIBentrySTDinterwordspacing}{\spaceskip=0pt\relax}
\providecommand{\BIBentryALTinterwordstretchfactor}{4}
\providecommand{\BIBentryALTinterwordspacing}{\spaceskip=\fontdimen2\font plus
\BIBentryALTinterwordstretchfactor\fontdimen3\font minus
  \fontdimen4\font\relax}
\providecommand{\BIBforeignlanguage}[2]{{%
\expandafter\ifx\csname l@#1\endcsname\relax
\typeout{** WARNING: IEEEtran.bst: No hyphenation pattern has been}%
\typeout{** loaded for the language `#1'. Using the pattern for}%
\typeout{** the default language instead.}%
\else
\language=\csname l@#1\endcsname
\fi
#2}}
\providecommand{\BIBdecl}{\relax}
\BIBdecl

\bibitem{Jackson:08}
M.~Jackson, \emph{Social and Economic Networks}.\hskip 1em plus 0.5em minus
  0.4em\relax Princeton University Press, 2008.

\bibitem{Newman:10}
M.~Newman, \emph{Networks: An Introduction}.\hskip 1em plus 0.5em minus
  0.4em\relax Oxford University Press, 2010.

\bibitem{SandryhailaM:13}
A.~Sandryhaila and J.~M.~F. Moura, ``Discrete signal processing on graphs,''
  \emph{IEEE Trans. Signal Process.}, vol.~61, no.~7, pp. 1644--1656, Apr.
  2013.

\bibitem{SandryhailaM:131}
------, ``Discrete signal processing on graphs: {F}requency analysis,''
  \emph{IEEE Trans. Signal Process.}, vol.~62, no.~12, pp. 3042--3054, Jun.
  2014.

\bibitem{ShumanNFOV:13}
D.~I. Shuman, S.~K. Narang, P.~Frossard, A.~Ortega, and P.~Vandergheynst, ``The
  emerging field of signal processing on graphs: {E}xtending high-dimensional
  data analysis to networks and other irregular domains,'' \emph{IEEE Signal
  Process. Mag.}, vol.~30, pp. 83--98, May 2013.

\bibitem{HammondVG:11}
D.~K. Hammond, P.~Vandergheynst, and R.~Gribonval, ``Wavelets on graphs via
  spectral graph theory,'' \emph{Appl. Comput. Harmon. Anal.}, vol.~30, pp.
  129--150, Mar. 2011.

\bibitem{NarangO:2012}
S.~K. Narang and A.~Ortega, ``Perfect reconstruction two-channel wavelet filter
  banks for graph structured data,'' \emph{IEEE Trans. Signal Process.},
  vol.~60, no.~6, pp. 2786--2799, Jun. 2012.

\bibitem{NarangO:2013}
------, ``Compact support biorthogonal wavelet filterbanks for arbitrary
  undirected graphs,'' \emph{IEEE Trans. Signal Process.}, vol.~61, no.~19, pp.
  4673--4685, Oct. 2013.

\bibitem{NarangSO:10}
S.~K. Narang, G.~Shen, and A.~Ortega, ``Unidirectional graph-based wavelet
  transforms for efficient data gathering in sensor networks,'' in \emph{Proc.
  IEEE Int. Conf. Acoust., Speech Signal Process.}, Dallas, TX, Mar. 2010, pp.
  2902--2905.

\bibitem{Pesenson:08}
I.~Z. Pesenson, ``Sampling in {P}aley-{W}iener spaces on combinatorial
  graphs,'' \emph{Trans. Amer. Math. Soc.}, vol. 360, no.~10, pp. 5603--5627,
  May 2008.

\bibitem{NarangGO:13}
S.~K. Narang, A.~Gadde, and A.~Ortega, ``Signal processing techniques for
  interpolation in graph structured data,'' in \emph{Proc. IEEE Int. Conf.
  Acoust., Speech Signal Process.}, Vancouver, May 2013, pp. 5445--5449.

\bibitem{WangLG:14}
X.~Wang, P.~Liu, and Y.~Gu, ``Local-set-based graph signal reconstruction,''
  \emph{IEEE Trans. Signal Process.}, 2014, submitted.

\bibitem{AgaskarL2013}
A.~Agaskar and Y.~M. Lu, ``A spectral graph uncertainty principle,'' \emph{IEEE
  Trans. Inf. Theory}, vol.~59, no.~7, pp. 4338 -- 4356, Jul. 2013.

\bibitem{ChenCRBGK:13}
S.~Chen, F.~Cerda, P.~Rizzo, J.~Bielak, J.~H. Garrett, and J.~Kova{\v
  c}evi{\'c}, ``Semi-supervised multiresolution classification using adaptive
  graph filtering with application to indirect bridge structural health
  monitoring,'' \emph{IEEE Trans. Signal Process.}, vol.~62, no.~11, pp. 2879
  -- 2893, Jun. 2014.

\bibitem{SandryhailaM:13g}
A.~Sandryhaila and J.~M.~F. Moura, ``Classification via regularization on
  graphs,'' in \emph{IEEE GlobalSIP}, Austin, TX, Dec. 2013, pp. 495--498.

\bibitem{EkambaramFAB:13}
V.~N. Ekambaram, B.~A. G.~Fanti, and K.~Ramchandran, ``Wavelet-regularized
  graph semi-supervised learning,'' in \emph{IEEE GlobalSIP}, Austin, TX, Dec.
  2013, pp. 423 -- 426.

\bibitem{ZhangDF:12}
X.~Zhang, X.~Dong, and P.~Frossard, ``Learning of structured graph
  dictionaries,'' in \emph{Proc. IEEE Int. Conf. Acoust., Speech Signal
  Process.}, Kyoto, Japan, 2012, pp. 3373--3376.

\bibitem{ThanouSF:13}
D.~Thanou, D.~I. Shuman, and P.~Frossard, ``Parametric dictionary learning for
  graph signals,'' in \emph{IEEE GlobalSIP}, Austin, TX, Dec. 2013, pp.
  487--490.

\bibitem{ChenO:14}
P.-Y. Chen and A.~Hero, ``Local {F}iedler vector centrality for detection of
  deep and overlapping communities in networks,'' in \emph{Proc. IEEE Int.
  Conf. Acoust., Speech Signal Process.}, Florence, Italy, 2014, pp. 1120 --
  1124.

\bibitem{SandryhailaM:14}
A.~Sandryhaila and J.~M.~F. Moura, ``Big data processing with signal processing
  on graphs,'' \emph{IEEE Signal Process. Mag.}, vol.~31, no.~5, pp. 80 -- 90,
  2014.

\bibitem{ElmoatazLB:08}
A.~Elmoataz, O.~Lezoray, and S.~Bougleux, ``Nonlocal discrete regularization on
  weighted graphs: A framework for image and manifold processing,'' \emph{IEEE
  Trans. Image Process.}, vol.~17, no.~7, pp. 1047--1060, Jul. 2008.

\bibitem{ChenSMK:13}
S.~Chen, A.~Sandryhaila, J.~M.~F. Moura, and J.~Kova{\v c}evi{\'c}, ``Adaptive
  graph filtering: {M}ultiresolution classification on graphs,'' in \emph{IEEE
  GlobalSIP}, Austin, TX, Dec. 2013, pp. 427 -- 430.

\bibitem{DongFVN:14}
X.~Dong, P.~Frossard, P.~Vandergheynst, and N.~Nefedov, ``Clustering on
  multi-layer graphs via subspace analysis on {G}rassmann manifolds,''
  \emph{IEEE Trans. Signal Process.}, vol.~62, no.~4, pp. 905--918, Feb. 2014.

\bibitem{Chung:96}
F.~R.~K. Chung, \emph{Spectral Graph Theory (CBMS Regional Conference Series in
  Mathematics, No. 92)}.\hskip 1em plus 0.5em minus 0.4em\relax Am. Math. Soc.,
  1996.

\bibitem{BelkinN:03}
M.~Belkin and P.~Niyogi, ``Laplacian eigenmaps for dimensionality reduction and
  data representation,'' \emph{Neur. Comput.}, vol.~13, pp. 1373--1396, 2003.

\bibitem{PueschelM:08}
M.~P{\"u}schel and J.~M.~F. Moura, ``Algebraic signal processing theory:
  Foundation and {1-D} time,'' \emph{IEEE Trans. Signal Process.}, vol.~56,
  no.~8, pp. 3572--3585, Aug. 2008.

\bibitem{Pueschelm:08b}
------, ``Algebraic signal processing theory: {1-D} space,'' \emph{IEEE Trans.
  Signal Process.}, vol.~56, no.~8, pp. 3586--3599, Aug. 2008.

\bibitem{Mallat:09}
S.~Mallat, \emph{A Wavelet Tour of Signal Processing}, 3rd~ed.\hskip 1em plus
  0.5em minus 0.4em\relax New York, NY: Academic Press, 2009.

\bibitem{BuadesCM:05}
A.~Buades, B.~Coll, and J.~M. Morel, ``A review of image denoising algorithms,
  with a new one,'' \emph{Multiscale Modeling \& Simulation}, vol.~4, pp.
  490--530, Jul. 2005.

\bibitem{Rudin:92}
L.~I. Rudin, Osher, and E.~Fatemi, ``Nonlinear total variation based noise
  removal algorithms,'' \emph{Physica D}, vol.~60, no. 1--4, pp. 259--268, Nov.
  1992.

\bibitem{Chan:01}
T.~F. Chan, S.~Osher, and J.~Shen, ``The digital {TV} filter and nonlinear
  denoising,'' \emph{IEEE Trans. Image Process.}, vol.~10, no.~2, pp. 231--241,
  Feb. 2001.

\bibitem{ChambolleA:04}
A.~Chambolle, ``An algorithm for total variation minimization and
  applications,'' \emph{J. Math. Imag. Vis.}, vol.~20, no. 1-2, pp. 89--97,
  Jan. 2004.

\bibitem{Donoho:06}
D.~L. Donoho, ``Compressed sensing,'' \emph{IEEE Trans. Inf. Theory}, vol.~52,
  no.~4, pp. 1289--1306, Apr. 2006.

\bibitem{CandesRT:06a}
E.~J. Cand{\'e}s, J.~K. Romberg, and T.~Tao, ``Stable signal recovery from
  incomplete and inaccurate measurements,'' \emph{Comm. Pure Appl. Math.},
  vol.~59, pp. 1207--1223, Aug. 2006.

\bibitem{CandesR:09}
E.~J. Cand{\'e}s and B.~Recht, ``Exact matrix completion via convex
  optimization,'' \emph{Journal Foundations of Computational Mathematics.},
  vol.~9, no.~2, pp. 717--772, Dec. 2009.

\bibitem{CandesP:10}
E.~J. Cand{\'e}s and Y.~Plan, ``Matrix completion with noise,''
  \emph{Proceedings of the IEEE}, vol.~98, pp. 925--936, Jun. 2010.

\bibitem{CandesLMW:11}
E.~J. Cand{\'e}s, X.~Li, Y.~Ma, and J.~Wright, ``Robust principal component
  analysis?'' \emph{Journal of the ACM}, vol.~58, May 2011.

\bibitem{WrightGRM:09}
J.~Wright, A.~Ganesh., S.~Rao, Y.~Peng, and Y.~Ma, ``Robust principal component
  analysis: Exact recovery of corrupted low-rank matrices by convex
  optimization,'' in \emph{Proc. Neural Information Process. Syst.}, Dec. 2009,
  pp. 2080--2088.

\bibitem{ChanS:05}
T.~F. Chan and J.~Shen, ``Variational image inpainting,'' \emph{Comm. Pure
  Applied Math}, vol.~58, pp. 579--619, Feb. 2005.

\bibitem{MairalES:08}
J.~Mairal, M.~Elad, and G.~Sapiro, ``Sparse representation for color image
  restoration,'' \emph{IEEE Trans. Image Process.}, vol.~17, pp. 53--69, Jan.
  2008.

\bibitem{KeshavanMO:10}
R.~H. Keshavan, A.~Montanari, and S.~Oh, ``Matrix completion from noisy
  entries,'' \emph{J. Machine Learn. Research}, vol.~11, pp. 2057--2078, Jul.
  2010.

\bibitem{MardaniMG:13}
M.~Mardani, G.~Mateos, and G.~B. Giannakis, ``Decentralized
  sparsity-regularized rank minimization: Algorithms and applications,''
  \emph{IEEE Trans. Signal Process.}, vol.~61, pp. 5374 -- 5388, Nov. 2013.

\bibitem{AnisGO:14}
A.~Anis, A.~Gadde, and A.~Ortega, ``Towards a sampling theorem for signals on
  arbitrary graphs,'' in \emph{Proc. IEEE Int. Conf. Acoust., Speech Signal
  Process.}, May 2014, pp. 3864--3868.

\bibitem{ZhuGL:03}
X.~Zhu, Z.~Ghahramani, and J.~Lafferty, ``Semi-supervised learning using
  {G}aussian fields and harmonic functions,'' in \emph{Proc. ICML}, 2003, pp.
  912--919.

\bibitem{ZhouS:04}
D.~Zhou and B.~Scholkopf, ``A regularization framework for learning from graph
  data,'' in \emph{ICML Workshop Stat. Rel. Learn.}, 2004, pp. 132--137.

\bibitem{BelkinNS:06}
M.~Belkin, P.~Niyogi, and P.~Sindhwani, ``Manifold regularization: A geometric
  framework for learning from labeled and unlabeled examples.'' \emph{J.
  Machine Learn. Research}, vol.~7, pp. 2399--2434, 2006.

\bibitem{ChenSLWMRBGK:14}
S.~Chen, A.~Sandryhaila, G.~Lederman, Z.~Wang, J.~M.~F. Moura, P.~Rizzo,
  J.~Bielak, J.~H. Garrett, and J.~Kova{\v c}evi{\'c}, ``Signal inpainting on
  graphs via total variation minimization,'' in \emph{Proc. IEEE Int. Conf.
  Acoust., Speech Signal Process.}, Florence, Italy, May 2014, pp. 8267 --
  8271.

\bibitem{ChenSMK:14a}
S.~Chen, A.~Sandryhaila, J.~M.~F. Moura, and J.~Kova{\v c}evi{\'c}, ``Signal
  denoising on graphs via graph filtering,'' in \emph{Proc. IEEE Glob. Conf.
  Signal Information Process.}, Atlanta, GA, Dec. 2014.

\bibitem{VetterliKG:12}
\BIBentryALTinterwordspacing
M.~Vetterli, J.~Kova{\v c}evi{\'c}, and V.~K. Goyal, \emph{Foundations of
  Signal Processing}.\hskip 1em plus 0.5em minus 0.4em\relax Cambridge
  University Press, 2014, http://www.fourierandwavelets.org/. [Online].
  Available: \url{http://www.fourierandwavelets.org/}
\BIBentrySTDinterwordspacing

\bibitem{DonohoE:03}
D.~Donoho and M.~Elad, ``Optimally sparse representation in general
  (nonorthogonal) dictionaries via $\ell^1$ minimization,'' \emph{Proc. Nat.
  Acad. Sci.}, vol. 100, no.~5, pp. 2197--2202, Mar. 2003.

\bibitem{CandesP:09}
E.~J. Cand{\`e}s and Y.~Plan, ``Near-ideal model selection by $\ell^1$
  minimization,'' \emph{Ann. Statist.}, vol.~37, no.~5A, pp. 2145--2177, 2009.

\bibitem{BoydPCPE:11}
S.~Boyd, N.~Parikh, E.~Chu, B.~Peleato, and J.~Eckstein, ``Distributed
  optimization and statistical learning via the alternating direction method of
  multipliers,'' \emph{Found. Trends Mach. Learn.}, vol.~3, no.~1, pp. 1--122,
  Jan. 2011.

\bibitem{BoydV:04}
S.~Boyd and L.~Vandenberghe, \emph{Convex Optimization}.\hskip 1em plus 0.5em
  minus 0.4em\relax New York, NY, USA: Cambridge University Press, 2004,
  vol.~17, no.~7.

\bibitem{Zhu:05}
X.~Zhu, ``Semi-supervised learning literature survey,'' Univ.
  Wisconsin-Madison, Tech. Rep. 1530, 2005.

\bibitem{AdamicG:05}
L.~A. Adamic and N.~Glance, ``The political blogosphere and the 2004 {U}.{S}.
  election: Divided they blog,'' in \emph{Proc. LinkKDD}, 2005, pp. 36--43.

\bibitem{CerdaGBRBZCM:12}
F.~Cerda, J.~Garrett, J.~Bielak, P.~Rizzo, J.~A. Barrera, Z.~Zhang, S.~Chen,
  M.~T. McCann, and J.~Kova{\v c}evi{\'c}, ``Indirect structural health
  monitoring in bridges: scale experiments,'' in \emph{Proc. Int. Conf. Bridge
  Maint., Safety Manag.}, Lago di Como, Jul. 2012, pp. 346--353.

\bibitem{CerdaCBGRK:12}
F.~Cerda, S.~Chen, J.~Bielak, J.~H. Garrett, P.~Rizzo, and J.~Kova{\v
  c}evi{\'c}, ``Indirect structural health monitoring of a simplified
  laboratory-scale bridge model,'' \emph{Int. J. Smart Struct. Syst., Sp. Iss.
  Challenge on bridge health monitoring utilizing vehicle-induced vibrations},
  vol.~13, no.~5, May 2013.

\bibitem{LedermanWBNGCKCR:14}
G.~Lederman, Z.~Wang, J.~Bielak, H.~Noh, J.~H. Garrett, S.~Chen, J.~Kova{\v
  c}evi{\'c}, F.~Cerda, and P.~Rizzo, ``Damage quantification and localization
  algorithms for indirect {SHM} of bridges,'' in \emph{Proc. Int. Conf. Bridge
  Maint., Safety Manag.}, Shanghai, Jul. 2014, pp. 640 -- 647.

\bibitem{Jester}
\BIBentryALTinterwordspacing
J.~jokes. Http://eigentaste.berkeley.edu/user/index.php. [Online]. Available:
  \url{http://eigentaste.berkeley.edu/user/index.php}
\BIBentrySTDinterwordspacing

\bibitem{MazumderHT:10}
R.~Mazumder, T.~Hastie, and R.~Tibshirani, ``Spectral regularization algorithms
  for learning large incomplete matrices,'' \emph{J. Mach. Learn. Res.},
  vol.~11, pp. 2287--2322, Aug. 2010.

\bibitem{CaiCS:10}
J.-F. Cai, E.~J. Cand\`{e}s, and Z.~Shen, ``A singular value thresholding
  algorithm for matrix completion,'' \emph{SIAM J. on Optimization}, vol.~20,
  no.~4, pp. 1956--1982, Mar. 2010.

\bibitem{ZhangWM:06}
S.~Zhang, W.~Wang, J.~Ford, and F.~Makedon, ``Learning from incomplete ratings
  using non-negative matrix factorization,'' in \emph{In Proc. of SIAM
  Conference on Data Mining (SDM)}, 2006, pp. 549--553.

\bibitem{GuZD:10}
Q.~Gu, J.~Zhou, and C.~Ding, ``Collaborative filtering: Weighted nonnegative
  matrix factorization incorporating user and item graphs,'' in \emph{In Proc.
  of SIAM Conference on Data Mining (SDM)}, 2010, pp. 199--210.

\bibitem{KeshavanMO:09}
R.~H. Keshavan, A.~Montanari, and S.~Oh, ``Low-rank matrix completion with
  noisy observations: A quantitative comparison,'' in \emph{Proceedings of the
  47th Annual Allerton Conference on Communication, Control, and Computing},
  ser. Allerton'09, 2009, pp. 1216--1222.

\bibitem{KuruvillaSHK:13}
A.~Kuruvilla, N.~Shaikh, A.~Hoberman, and J.~Kova{\v c}evi{\'c}, ``Automated
  diagnosis of otitis media: {A} vocabulary and grammar,'' \emph{Int. J.
  Biomed. Imag., Sp. Iss. Computer Vis. Image Process. for Computer-Aided
  Diagnosis}, Aug. 2013.

\bibitem{CholletiGBPDSP:09}
S.~R. Cholleti, S.~A. Goldman, A.~Blum, D.~G. Politte, S.~Don, K.~Smith, and
  F.~Prior, ``Veritas: Combining expert opinions without labeled data.''
  \emph{International Journal on Artificial Intelligence Tools}, vol.~18,
  no.~5, pp. 633--651, 2009.

\bibitem{ChenSMK:14:web}
S.~Chen, A.~Sandryhaila, J.~M.~F. Moura, and J.~Kova{\v c}evi{\'c}. (2015)
  Signal recovery on graphs: Variation minimization.
  Http://jelena.ece.cmu.edu/\-repository/\-rr/\-15\_ChenSMK/\-15\_ChenSMK.html.

\end{thebibliography}

\section{Appendix}

\label{sec:appendix}
To decompose the
graph total variation term and nuclear norm term, we introduce an auxiliary
matrix $\Z$ to duplicate $\X$ and a residual matrix
$\Co$ to capture the error introduced by the inaccessible part of
$\T$ in~\eqref{eq:Measurement}, and then rewrite~\eqref{eq:C_Opt} in the equivalent form
\begin{eqnarray*}
\nonumber
\Xw, \Ww, \Ew, \Zw, \Cw &=& \arg \min_{\X, \W, \E, \Z, \Co}\ \left\|\W \right\|_F^2+\alpha \STV_2(\Z)\\
\label{eq:Opt_AL}
&&+ \beta \left\|\X\right\|_* + \gamma \left\|\E\right\|_1 + \I(\Co_\M),\\
\nonumber
\text{subject to}
&& \T = \X + \W + \E + \Co
\\
\label{eq:Opt_AL_cond}
&& \Z = \X,
\end{eqnarray*}
where $\I$ is an indicator operator defined as
\begin{equation*}
\label{eq:indicator}
\I(\X)_{n,m} = \begin{cases}
0,& \text{if \,\,\,} \X_{n,m} = 0, \\
+\infty,& \text{otherwise}.
\end{cases}
\end{equation*}
Note that we move $\left\|\W\right\|_F^2$ from the constraint to the
objective function and putting $\I(\Co_\M)$ in the objective function is equivalent to setting $\Co_\M$ to be zero. We then construct the
augmented Lagrangian function
\begin{eqnarray*}
&& \L{(\X,\W,\E,\Z,\Co,\Y_1,\Y_2)} \\
\nonumber
&=& \left\|\W\right\|_F^2 + \alpha \STV_2{(\Z)} +  \beta \left\|\X\right\|_*+ \gamma \left\|\E\right\|_1 + \I{(\Co_\M)} \\
\nonumber
&& + < \Y_1, \T-\X-\W-\E-\Co > + <\Y_2, \X - \Z > \\
\nonumber
&& + \frac{\eta}{2} \left\| \T-\X-\W-\E-\Co\right\|_F^2 +\frac{\eta}{2} \left\|\X - \Z \right\|_F^2.
\end{eqnarray*}
We minimize for each variable individually. For $\X$, we aim to solve
\begin{eqnarray*}
\min_{\X}   \beta \left\|\X\right\|_* + \frac{\eta}{2} \left\| \T-\X-\W-\E-\Co - \Y_1 \right\|_F^2 +\frac{\eta}{2} \left\|\X - \Z - \Y_2\right\|_F^2.
\end{eqnarray*}
We solve it by the standard matrix-soft thresholding~\eqref{eq:shrinkage_singular}; for $\W$,  we aim to solve
\begin{eqnarray*}
\min_{\W}   \left\|\W\right\|_F^2 + \frac{\eta}{2} \left\| \T-\X-\W-\E-\Co - \Y_1 \right\|_F^2.
\end{eqnarray*}
We obtain the closed-form solution by setting the derivative to zero; for $\E$,  we aim to solve
\begin{eqnarray*}
\min_{\E}  \gamma  \left\|\E\right\|_1 + \frac{\eta}{2} \left\| \T-\X-\W-\E-\Co - \Y_1 \right\|_F^2.
\end{eqnarray*}
We solve it by standard soft thresholding~\eqref{eq:shrinkage}; for $\Z$, we aim to solve
\begin{eqnarray*}
\min_{\X}   \alpha \left\|\Z - \Adj \Z \right\|_F^2 +\frac{\eta}{2} \left\|\X - \Z - \Y_2\right\|_F^2.
\end{eqnarray*}
We obtain the closed-form solution by setting the derivative to zero. For $\Co$, we just set $\Co_M = 0$ to satisfy the constraints; for $\Y_1$ and $\Y_2$, we update them as the standard Lagrange multipliers. For the final implementation see Algorithm~\ref{alg:GSR}. 

\end{document}